\documentclass[10pt,notitlepage,aps,pra,superscriptaddress,floatfix,footinbib,twocolumn]{revtex4-2}
\usepackage[english]{babel}
\usepackage[T1]{fontenc}
\usepackage[utf8]{inputenc}
\usepackage{lipsum}
\usepackage{comment}
\usepackage{amsmath, amsthm, amssymb, mathtools, amsfonts}
\usepackage{physics}

\newcommand{\beq}{\begin{equation}}
\newcommand{\eeq}{\end{equation}}

\newcommand{\set}[1]{\mathcal{#1}}

\newcommand{\rr}{\mathbf{r}}
\newcommand{\Var}{\operatorname{Var}}
\newcommand{\EE}{\mathbb E}

\newtheorem{thm}{Theorem}
\newtheorem{prop}[thm]{Proposition}
\newtheorem{coro}[thm]{Corollary}

\newtheorem{lemma}[thm]{Lemma}
\theoremstyle{definition}
\newtheorem{obs}[thm]{Observation}
\newtheorem{dfn}[thm]{Definition}
\usepackage[final]{graphicx}
\usepackage{booktabs}
\usepackage[dvipsnames]{xcolor}
\usepackage{tikz}
\usetikzlibrary{arrows.meta, shapes, decorations.pathmorphing, calc}
\usepackage{pgfplots}
\usepgfplotslibrary{fillbetween}
\pgfplotsset{compat=1.15}
\definecolor{telecolor}{rgb}{0.7099,0.0471,0.2118}
\definecolor{unipicolor}{RGB}{14,75,124}
\definecolor{unipicolorlight}{RGB}{24, 126, 209}
\definecolor{greencolor}{RGB}{46,111,64}

\usepackage{bm}
\usepackage{hyperref}
 \hypersetup{
     colorlinks=true,
     linkcolor=unipicolor,
     filecolor=teal,
     citecolor=unipicolor,
     draft=false
     }
\usepackage{bookmark}
\let\inner\relax
\NewDocumentCommand\inner{mg}{%
	\ensuremath{\left\langle #1| \IfNoValueTF{#2}{#1}{#2}\right\rangle}%
}
\let\outer\relax
\NewDocumentCommand\outer{mg}{%
	\ensuremath{\ket{#1}\!\!\IfNoValueTF{#2}{\bra{#1}}{\bra{#2}}}%
}
\begin{document}
\title{Bounding single-copy information causality violations}
\author{Marco Pompili}
\affiliation{Inria, Télécom Paris - LTCI, Institut Polytechnique de Paris, 91120 Palaiseau, France}
\affiliation{Department of Physics ``E.~Fermi'', University of Pisa, Largo B.~Pontecorvo 3, 56127 Pisa, Italy}
\author{Costantino Budroni}
\affiliation{Department of Physics ``E.~Fermi'', University of Pisa, Largo B.~Pontecorvo 3, 56127 Pisa, Italy}
\date{\today}

\begin{abstract}
    Information causality is a  physical principle
    that constrains bipartite post-quantum nonlocal correlations
    from information-theoretic assumptions. A direct
    characterization of the correlations satisfying it, however, remains difficult because
    the defining entropic condition must hold for all admissible protocols
    in the corresponding task. Here, we formulate the single-copy
    information causality problem as an optimization over
    all protocol parameters and simplify it by proving a deterministic reduction.
    We then show that generic adaptive multi-copy 
    protocols can be represented as
    single-copy protocols applied to an effective behavior obtained by local
    wirings. As a consequence, for families of correlations closed under
    wirings, excluding single-copy violations 
    suffices to exclude all finite-copy adaptive-protocol violations.
    We formulate the resulting membership and set-inclusion problems as
    optimization problems and discuss possible semidefinite programming relaxations.
    As an application, 
    we study expansions of the information causality functional around zero-capacity channels,
    generalizing a method recently employed by Gachechiladze
    and Miklin to derive the Tsirelson--Landau--Masanes boundary from
    information causality. For expansions centered on channels whose output 
    distribution has full support, imposing the second-order
    constraints for all single-copy protocols characterizes exactly the set of macroscopic locality 
    in the CHSH scenario.
    Expanding instead around appropriate constant-output channels with perturbations
    that change the channel support, we show that every 
    single-copy IC violation can be witnessed already at first order with a channel of 
    arbitrarily small capacity.
\end{abstract}

\maketitle

\section{Introduction}

Bell nonlocality shows that quantum correlations cannot be reproduced by local
hidden-variable models~\cite{bell:1964}. At the same time, the
no-signaling (NS) principle alone allows correlations that are strictly
stronger than those achievable in quantum theory \cite{popescu:1994,
brunner:2014}. This observation motivated the search for operational principles
that constrain the set of nonsignaling behaviors without assuming the
Hilbert-space formalism of quantum mechanics \cite{vandam:2013, brassard:2006,
linden:2007, navascues:2010, fritz:2013, cabello:2013}.

Information causality (IC) is one of the most prominent proposals in this
program~\cite{pawlowski:2009}. In its standard formulation, Alice and Bob
perform a specific distributed computing
protocol, namely a
\textit{random access code assisted by a shared nonsignaling resource}~\cite{ambainis:1999,ambainis2002,pawlowski:2010}, and the
information potentially available to Bob about
the entries of a database held by Alice is required not to
exceed the information transmitted to him through  a
classical channel. Both classical and quantum theories satisfy this requirement.
The assumption of IC as a physical principle allows one to recover bounds for
the set of correlations that coincide with the quantum ones, such as the
Tsirelson bound~\cite{cirelson:1980} for the Clauser--Horne--Shimony--Holt Bell
inequality~\cite{clauser:1969}, the Uffink inequality \cite{uffink:2002}, and
several portions of the quantum set of correlations \cite{pawlowski:2009,
allcock:2009, gachechiladze:2022, jain:2024, gachechiladze:2026}. Despite these
successes, a direct characterization of the correlations satisfying IC has
remained elusive: the defining constraint is entropic and nonlinear, and it must
hold for every admissible encoding and decoding protocol.

Several other physical or information-theoretic principles have been proposed,
including macroscopic locality~\cite{navascues:2010}, non-trivial communication
complexity~\cite{vandam:2013, brassard:2006}, no advantage for nonlocal
computation \cite{linden:2007}, local orthogonality~\cite{fritz:2013}, and the
E-principle \cite{cabello:2013}. These principles are all satisfied by the
almost-quantum set $\set{AQ}$~\cite{navascues:2015}, which is strictly larger
than the quantum set. Information causality is the notable unresolved case:
neither a proof that all almost-quantum correlations satisfy information
causality nor a counterexample is known,
while some correlations outside the almost-quantum set satisfy IC
for all protocols~\cite{pompili:2026}.

Recent work has substantially enlarged the range of tractable IC constraints.
The noisy-channel formulation by Miklin and Paw\l{}owski \cite{miklin:2021}
showed that known  bounds,
originally derived from many \emph{concatenated} systems, can be
recovered, and in some cases improved, using a single copy of the shared
resource. This approach has since been extended to broader protocols and
scenarios~\cite{gachechiladze:2022, pollyceno:2023, jain:2024,
jain:2026extended, miklin:2026cc, pollyceno:2026}, while complementary
formulations of IC have also been investigated~\cite{yu:2026}. Recently, the
extended single-copy version of IC with correlated entries~\cite{jain:2024,
miklin:2026cc} has been shown to imply the Tsirelson--Landau--Masanes (TLM)
inequality \cite{tsirelson:1987, landau:1988, masanes:2003}, and consequently,
both the macroscopic locality \cite{navascues:2010} set,
as well as the exact correlator region for vanishing marginals,
in the simplest Bell scenario \cite{gachechiladze:2026}. 
These developments indicate that single-copy IC
constraints are considerably more powerful than previously expected, but they do
not provide a protocol-independent membership test for the standard
random-access-code formulation, nor do they settle its relation with $\set{AQ}$.

In this work, we systematize single- and multi-copy IC protocols and
    establish reductions that simplify their analysis.
We represent a general single-copy protocol, within the causal structure
considered here, by a database distribution and by stochastic maps
describing the parties' measurement choices, the noisy classical communication,
and Bob's decoding. Convexity arguments allow 
the measurement and decoding maps to be taken deterministic and lead
naturally to  single-copy IC sets
in each Bell scenario. We then show that finite-copy adaptive protocols,
in which each party chooses both the order in which it uses its copies and the
corresponding inputs as a function of its local data, can be represented as
single-copy protocols applied to an effective behavior obtained by local
wirings. Consequently, for any family of correlations closed under wirings
\cite{allcock:2009closed, lang:2014, Beigi2015}, a finite-copy violation implies a
single-copy violation by an effective behavior belonging to the same family. In
particular, if the almost-quantum family violates information causality through
any finite-copy protocol considered here, then a violation can already be found
at the single-copy level. Despite these reductions, optimizing the value of the IC
expression is still a nontrivial problem. We discuss possible methods to approximate its
values with lower and upper bounds, obtained by numerical methods or semidefinite programming
relaxations. Finally, as an application of our results, we discuss perturbative expansions
of IC around zero-capacity channels. In the case of perturbations that do not change the channel
support, we recover the recent result 
of Gachechiladze and Miklin \cite{gachechiladze:2026}, but this time as an 
optimization over all possible IC single-copy protocols, instead of choosing a particular 
one. Instead, expanding around a channel with constant output, but allowing the perturbation
to change the channel support, we show that single-copy IC violations can be 
witnessed already at first order with a channel of arbitrarily small capacity. This answers
the question, raised by Gachechiladze and Miklin, of whether 
their bound could be improved at leading order in the vanishing capacity limit
with a better choice of the IC protocol: it is possible when the perturbation changes the channel
support and impossible, in the CHSH scenario, when it does not.

\section{Preliminary notions}\label{sec:preliminary}

We consider an assisted random access coding (RAC) task, as in the original
proposal~\cite{pawlowski:2009}, with imperfect communication between the
parties~\cite{miklin:2021}. Alice receives a string of dits
$\mathbf R = (R_0, \dots, R_{n-1})$, each taking values in $[d]$, and
Bob independently receives an index $S$ taking values in $[n]$,
where, for $k \in \mathbb N$, we use the notation
$[k] \coloneqq \{0, \dots, k - 1\}$. Bob's goal is to guess the value of
$R_s$ upon receiving $S = s$. Bob's guess is denoted by
$G$ (which, in principle, could take values in any finite alphabet).
Alice sends a message $M$ through a noisy classical channel, and Bob receives $M'$.
The players may use shared correlations $\{p(a,b|x,y)\}$ (also called
\emph{behavior} or \emph{box} throughout) belonging to a Bell scenario
$\mathsf B = (m_A, m_B, d_A, d_B)$, i.e., with $x \in [m_A]$, $y \in [m_B]$,
$a \in [d_A]$ and $b \in [d_B]$. We denote the RAC parameters as $(n,d,d')$,
where $M'$ takes values in $[d']$. This scenario is depicted in
Fig.~\ref{fig:ic_scenario}.

If the entries $R_0, \dots, R_{n-1}$ are independent, information causality states that
\beq
    \label{eq:ic_statement}
    \sum_{s=0}^{n-1} I(R_s : G_s) \le I(M:M'),
\eeq
where $I(R_s : G_s) = I(R_s : G \,|\, S = s)$ is the Shannon mutual information
between the target dit and the guess when $S = s$, and $I(M:M')$ quantifies the
amount of information communicated through the channel. Notice that the usual
formulation with the channel capacity $C := \sup_{p(m)} I(M:M')$ provides in
general a looser constraint. Equation~\eqref{eq:ic_statement} thus bounds the
information available to the receiver by the amount of communication exchanged
and reduces to the no-signaling (NS) principle~\cite{popescu:1994,brunner:2014} 
in the case of zero communication. The inequality holds whenever
the underlying theory admits a generalized mutual information
obeying some natural properties, such as non-negativity, a chain rule, reduction
to the Shannon mutual information for classical variables, a data-processing
inequality and a Markov property for sequential operations; see,
e.g.,~\cite{jain:2024} for a proof. Crucially, this includes quantum theory, so
every quantum correlation satisfies the inequality.

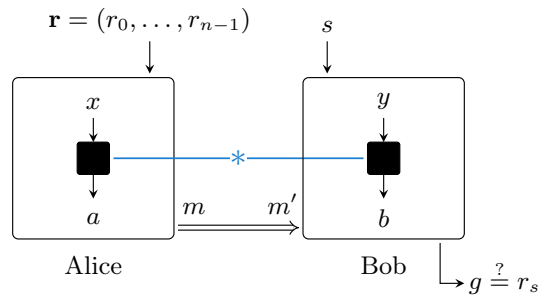
\begin{figure}
\centering
\resizebox{\columnwidth}{!}{\begin{tikzpicture}

        \useasboundingbox (-1.2, -0.55) rectangle (6.8, 2.9);
        
        \draw[rounded corners=2pt] (0,0) rectangle (2,2);
        \node[below] at (1,-0.1) {Alice};
        
        \draw[stealth-] (1.7,2.05) -- ++(0, 0.4) node[above=5pt, anchor=base] {$\mathbf r = (r_0, \dots, r_{n-1})$};
        
        \draw[fill=black, rounded corners=1pt] (0.8,0.8) rectangle (1.2, 1.2);
        \draw[stealth-] (1.0,1.2) -- ++(0, 0.3) node[above] {$x$};
        \draw[-stealth] (1.0,0.8) -- ++(0, -0.3) node[below=9pt, anchor=base] {$a$};
    
        \node[left=0.65cm, above=-0.25cm] at (0,-0.55) {\phantom{$g \stackrel{?}{=} r_s$}};
        
        \draw[baseline, double equal sign distance, -{Implies[]}, fill=black, black] (2.05,0.14) -- ++(1.5,0) 
        node[left=1.3cm, above=0.17cm, anchor=base] {$m$}
        node[left=0.2cm, above=0.17cm, anchor=base] {$m'$};
        
        \draw[-, color=unipicolorlight, semithick] (1.25,1.0) -- ++(3.1,0);
        \fill[white] (2.8,1.0) circle[radius=3.0pt]  node[unipicolorlight] {\large $\ast$};

        \begin{scope}[xshift=3.6cm]
        \draw[rounded corners=2pt] (0,0) rectangle (2,2);
        \node[below] at (1,-0.1) {Bob};
        \draw[stealth-] (0.3,2.05) -- ++(0, 0.4) node[above] {$s$};
        \draw[-stealth] (1.7,-0.05) -- ++(0, -0.5) -- ++(0.3,0) coordinate (G);
        \node[right=0.5cm, above=-0.25cm] at (G) {$g \stackrel{?}{=} r_s$};        
        
        \draw[fill=black, rounded corners=1pt] (0.8,0.8) rectangle (1.2, 1.2);
        \draw[stealth-] (1.0,1.2) -- ++(0, 0.3) node[above=5pt, anchor=base] {$y$};
        \draw[-stealth] (1.0,0.8) -- ++(0, -0.3) node[below=10pt, anchor=base] {$b$};
        \end{scope}
\end{tikzpicture}}
\caption{Notation used for the random access coding task where
information causality is defined. Bob's task is to
produce a guess $g$ equal to $r_s$ upon receiving the input $s$,
aided by the imperfect classical communication $m \to m'$ and
(possibly multiple copies of)
the shared correlations $\{p(a,b|x,y)\}$.}
\label{fig:ic_scenario}
\end{figure}

In the Clauser-Horne-Shimony-Holt (CHSH) scenario with all variables binary, it
is possible to show that all correlations beyond Tsirelson's bound for the CHSH
correlator violate the IC inequality~\cite{pawlowski:2009, miklin:2021} by
considering $n = 2$
independent bits, a binary symmetric channel of sufficiently small
    positive capacity,
and the van Dam protocol,
\beq
    \label{eq:vandam}
    x = r_{0} \oplus r_{1}, \quad y
      = s, \quad m
      = a \oplus r_{0}, \quad g
      = b \oplus m',
\eeq
originally introduced in the study of communication
complexity~\cite{vandam:2013}. This result provides an information-theoretic
explanation for Tsirelson's bound.
While the original strong IC bounds were obtained using concatenations of
several copies of the behavior and a noiseless message
$m$~\cite{pawlowski:2009}, the same conclusions, including Tsirelson's
bound, can already be obtained from a single copy together with a noisy
classical channel~\cite{miklin:2021}.
These results motivate our systematic single-copy analysis.

It is helpful to also recall the version of IC in the
generalized form~\cite{chaves:2015, jain:2024}
where the database entries are correlated~\footnote{Miklin et
al.~\cite{miklin:2026cc} discuss even more general protocols in which the RAC is
substituted by an arbitrary distributed computing task. We do not discuss
the most general case here, focusing on the usual RAC formulation of IC.}. 
Consider the same random-access-code setting as in the standard formulation
with
$\mathbf R = (R_0,\ldots,R_{n-1}) \in[d]^n$ distributed according to an
arbitrary joint distribution $q(\rr)$, while Bob receives an independent index
$S \in[n]$. We denote by $\set{D}_{n,d}$ the set of
database distributions allowed in the generalized IC scenario. Unless
stated otherwise, we take $\set{D}_{n,d}$ to be the full probability simplex,
i.e., we impose no restriction on such probabilities.
We fix the ordering $(R_0,\ldots,R_{n-1})$ and write
    $\mathbf R_{<s} := (R_0,\ldots,R_{s-1})$, with $\mathbf R_{<0}$ empty
    by definition.
    Bob's available information consists of the received message $M'$, his
    measurement choice $Y$, and his outcome $B$, from which he produces
    a guess $G$.
    We absorb Alice's encoding into an effective classical channel from
        $(\mathbf R,X,A)$ to $M'$. This yields a constraint at least as
        stringent as that in Eq.~\eqref{eq:ic_statement}, since data processing
        for the Markov chain $(\mathbf R,X,A) \to M \to M'$ implies
        $I(\mathbf R,X,A:M') \leq I(M:M')$. The extended IC inequality then
        reads
    
    \begin{equation}
        \sum_{s=0}^{n-1} I(R_s:G_s \mid \mathbf R_{<s}) \leq I(\mathbf R,X,A:M'). \label{eq:generalized_ic_statement}
    \end{equation}
    Here, each term on the left-hand side is conditioned on $S = s$.

    For independent database entries, the chain rule implies
        $I(R_s:G_s \mid \mathbf R_{<s}) \geq I(R_s:G_s)$.
    Thus, for independent database entries,
    Eq.~\eqref{eq:generalized_ic_statement} implies
    Eq.~\eqref{eq:ic_statement}.
    We discuss a natural multi-copy analogue of
        Eq.~\eqref{eq:generalized_ic_statement} in
        Sec.~\ref{sec:generalized_many_copy_ic}.

Finally, we briefly mention the almost-quantum set $\set{AQ}$, which is an outer
approximation of the quantum set
$\set{Q}$. The central properties for us are that it is closed under
wirings~\cite{navascues:2015} and
admits an efficient semidefinite characterization, coinciding with the
$1 + \mathrm{AB}$ level of the Navascués-Pironio-Acín (NPA) hierarchy
\cite{navascues:2007, navascues:2008}.

\section{General single-copy protocols}
\label{sec:generalized_single_copy_protocols}

In order to characterize the correlations satisfying
Eq.~\eqref{eq:generalized_ic_statement}, we consider  general
single-copy protocols. With respect to the standard approach,
 the database distribution
$q(\rr)$ is no longer fixed to the uniform product
distribution. Let
$P = \{p(a,b|x,y)\}$ be a nonsignaling behavior in a Bell scenario
$\mathsf B = (m_A,m_B,d_A,d_B)$. A general protocol is specified by the
stochastic maps
\begin{equation}
    \begin{aligned}
        p_X(x|\rr), \quad p_M(m'|\rr,x,a),          \\
        p_Y(y|s,m'), \quad p_G(g|s,m',y,b),
    \end{aligned}
\end{equation}
together with a database distribution $q \in \set{D}_{n,d}$. The
corresponding joint distribution of all variables is
\begin{equation}
    \begin{split}
        p(\rr,s,x,y,a,b,m',g) = q(\rr)p(s)p_X(x|\rr)p_Y(y|s,m') \\
        \times p_M(m'|\rr,x,a)p(a,b|x,y)p_G(g|s,m',y,b),
    \end{split}
    \label{eq:generalized_single_copy_protocol}
\end{equation}
where $\mathbf R$ and $S$ are independent and $p(s) = 1/n$ (for convenience,
as the quantities are conditioned on $S=s$). As in the standard
setting, nonsignaling allows Bob to choose $y$ after receiving $m'$ (see
Lemma~\ref{lem:sequential_marginals} in
Appendix~\ref{app:generalized_effective_ns} for $c = 1$).
Before producing his guess of $R_s$, Bob's available information is $(M',Y,B)$,
    together with $S = s$.
Note that the inclusion of $Y$ is relevant when $p_Y(y|s,m')$ is stochastic.
A general single-copy protocol is illustrated in Fig.~\ref{fig:dag_ic_full}.

We denote the maps by the corresponding column-stochastic matrices
$T_X,T_M,T_Y,T_G$ and write
\begin{equation}
    \pi := (T_X,T_M,T_Y,T_G)
\end{equation}
for the operational part of the protocol. For a fixed behavior $P$,
database distribution $q$, and protocol $\pi$, we define the generalized
single-copy IC deficit by
\begin{equation}
    \begin{aligned}
    \Delta_1(P,q,\pi) &\coloneqq \sum_{s=0}^{n-1}I(R_s:G_s \mid \mathbf R_{<s})\\
&\qquad-I(\mathbf R,X,A:M').
\end{aligned} \label{eq:generalized_single_copy_gap}
\end{equation}
The largest violation compatible with the chosen Bell and RAC scenarios is
therefore
\begin{equation}
    \Delta_1^{*}(P) := \sup_{\substack{q\in\set{D}_{n,d}\\ \pi\ \mathrm{compatible}}} \Delta_1(P,q,\pi). \label{eq:generalized_single_copy_optimization}
\end{equation}
Thus, in the generalized formulation the database distribution $q$ is
itself part of the optimization problem. The maximization in
Eq.~\eqref{eq:generalized_single_copy_optimization} is over both the continuous
parameters of $q$ and the stochastic protocol maps. When $\set{D}_{n,d}$
is the full simplex, all feasible sets are compact and the objective is
continuous, so the maximum is attained. It is useful to keep the role of $q$
explicit already at this stage, as in some arguments below, we may temporarily
fix $q$ while optimizing a particular transition matrix. We make this point
precise below.

\begin{figure}[t]
    \centering
    \begin{tikzpicture}[
    every path/.append style = {arrows = -stealth},
    visible/.style = {
        semithick,
    	draw = black, 
    	shape = circle, 
    	minimum size = 0.75cm,
    	shading = radial,
    	inner color = unipicolorlight!0,
    	outer color = unipicolorlight!0,
    },
    latent/.style = {
    	draw = none,
        shape = rectangle,
        inner sep = 2pt,
    	inner color = unipicolorlight!0,
    	outer color = unipicolorlight!0,
    	semithick
    }, 
    scale=0.9
    ]
    
    \useasboundingbox (-3,-2.5) rectangle (5.5,1.7);
    
    \begin{scope}[xshift=-0.5cm]
    
	
	\draw[-stealth] plot [smooth, tension=1.2] coordinates {(-0.28,0.90)  (-0.4,-1.8) (2.58,-2.4)};
	
	\node[visible] (x) at (0,0) {$X$};
    \node[visible] (y) at (3,0) {$Y$};
    \node[visible] (a) at (0,-1.2) {$A$};
    \node[visible] (b) at (3,-1.2) {$B$};
	\node[latent] (l) at (1.5, -0.5) {\color{unipicolorlight!90!black} $\omega_{AB}$};
	
	\node[visible, minimum width=0.8cm] (x0x1) at (0,1.2) {\phantom{$S$}};
	\node at (x0x1) {$\bm R$};
	\node[visible] (s) at (3,1.2) {$S$};
	\node[visible, anchor=base] (mp) at (3, -2.4) {\phantom{$M$}};
	\node at (mp) {$M'$};
	\node[visible] (g) at (4.2, -1.2) {$G$};
	
	
    \path (x) edge (a);
    \path[unipicolorlight!90!black] (l) edge (a);
    \path[unipicolorlight!90!black] (l) edge (b);
    \path (y) edge (b);

    \path (x0x1) edge (x);
    \path[bend right=20] (a) edge (mp);
    \path[bend right=20] (x) edge (mp);	
    \path (s) edge (y);
    \path (b) edge (g);
    \path (y) edge (g);
    
    \path[bend left=45] (mp) edge (y);

    
    \path[bend right=20] (mp) edge (g);
    \path[bend left=20] (s) edge (g);
    
    \end{scope}

\end{tikzpicture}
    \caption{
        A single-copy protocol for information causality, where the node
        $\omega_{AB}$ indicates that the variables $a,b$ conditioned on $x,y$
        are correlated via a nonsignaling behavior. Any other node depends on
        its parents via the corresponding transition matrix.
    }
    \label{fig:dag_ic_full}
\end{figure}
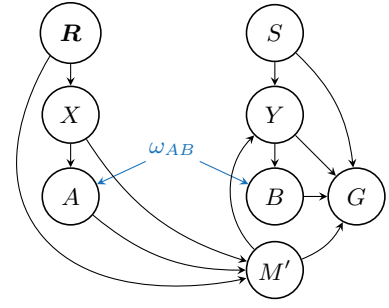

The optimization in Eq.~\eqref{eq:generalized_single_copy_optimization} can be
simplified using convexity properties of the expression in terms of the
distributions $p_X$, $p_Y$, and $p_G$. This result is stated formally
as follows.
\begin{lemma}[Single-copy deterministic reduction]
    \label{lemma:generalized_det_strat}
    Let $P$ be a fixed nonsignaling behavior and let $\set{D}_{n,d}$ be a fixed
    set of admissible database distributions. In the optimization
    problem of Eq.~\eqref{eq:generalized_single_copy_optimization}, it suffices
    to consider
    \begin{enumerate}
        \item deterministic choices of measurement for Alice,
              $p_X(x|\rr) \in \{0,1\}$;
        \item deterministic choices of measurement for Bob,
              $p_Y(y|s,m') \in \{0,1\}$;
        \item deterministic decoding for Bob,
                  $p_G(g|s,m',y,b) \in \{0,1\}$.
    \end{enumerate}
    No restriction is imposed on the database distribution
    $q \in \set{D}_{n,d}$ or on the communication map $T_M$.
\end{lemma}
The full proof is given in Appendix~\ref{app:generalized_deterministic_proofs}.
For fixed $q,T_M,T_Y,T_G$, the IC deficit is convex in $T_X$, so its
    maximum over stochastic maps is attained at a deterministic map. The same
    argument applies to $T_Y$ and $T_G$, after which one can optimize over $q$.

Denoting the deterministic measurement functions by $f_X$ and $f_Y$, we
have $x = f_X(\mathbf r)$ and $y = f_Y(s,m')$.
Since $Y$ is determined by $(S,M')$, it can be substituted into the
    decoding function, giving a deterministic function $g = f_G(s,m',b)$.
Moreover, since $X$ is a function of $\mathbf R$,
$I(\mathbf R,X,A:M') = I(\mathbf R,A:M')$. Consequently,
Eq.~\eqref{eq:generalized_single_copy_optimization} reduces to
\begin{equation}
        \begin{aligned}
            \Delta_1^{*}(P) = \max\ & \sum_{s=0}^{n-1}I(R_s: G_s \mid \mathbf R_{<s}) - I(\mathbf R,A:M') \\
            \text{s.t.}\            & q \in \set{D}_{n,d}, \quad [T_M]_{m',(\rr,a)}\ \text{stochastic},   \\
                                    & x = f_X(\rr), \quad y = f_Y(s,m'),                                  \\
                                    & g = f_G(s,m',b).
        \end{aligned}
        \label{eq:generalized_opt_prob_simplified}
    \end{equation}
 The optimization is therefore over a finite set of deterministic functions
$f_X,f_Y,f_G$, together with the continuous variables $q$ and $T_M$.

\section{Multi-copy information causality}
\label{sec:generalized_many_copy_ic}

\subsection{Adaptive protocols as wirings}
\label{sec:generalized_adaptive_wirings}

We now consider generalized IC protocols in which Alice and Bob share
$c \in \mathbb N$ independent copies of the same nonsignaling behavior
$P = \{p(a,b|x,y)\}$ and use them
adaptively, in the sense we now explain.
(The same construction and deterministic reduction apply to independent,
    possibly different nonsignaling behaviors $P_1,\ldots,P_c$.)
Alice receives the database register $\mathbf R$ distributed according
to an arbitrary $q \in \set{D}_{n,d}$, while Bob receives the independent index
$S$. Both parties use their copies sequentially. At every step, each of them
chooses, as a function of the local input and of everything observed so far,
both which unused copy to use next and which input to feed it. In particular,
the order in which  the parties use the copies may differ.
The only constraint relating the two parties is the one imposed by the
communication  step: Alice uses all her copies before the
message is sent, and Bob uses his only after receiving it.

\begin{dfn}[Sequential adaptive $c$-copy protocol]
    \label{dfn:sequential_adaptive_protocol}
    Label the copies by $1,\ldots,c$, and write
        $\mathbf X = (X_1,\ldots,X_c)$ and $\mathbf A = (A_1,\ldots,A_c)$ for
        Alice's inputs and outputs, and similarly $\mathbf Y,\mathbf B$ for Bob.
        At their ``time step'' $t \in \{1,\dots,c\}$, Alice and Bob choose copies
        $\ell_t$ and $j_t$,
        respectively.
        A general adaptive $c$-copy protocol consists of:
        \begin{enumerate}
            \item \emph{Alice's measurements.} At step $t = 1,\ldots,c$, Alice
                  chooses an unused copy $\ell_t$ and an input
                  $x_{\ell_t} \in[m_A]$, possibly stochastically, using $\rr$
                  and her previous choices and outcomes
                  $(\ell_u,x_{\ell_u},a_{\ell_u})_{u<t}$. She then obtains
                  $a_{\ell_t}$.
            \item \emph{Communication.} After Alice has used all copies, a
                  channel $p_M(m'|\rr,\boldsymbol\ell,\mathbf x,\mathbf a)$
                  produces the message received by Bob.
            \item \emph{Bob's measurements.} After receiving $m'$, at step
                  $t = 1,\ldots,c$ Bob chooses an unused copy $j_t$ and an input
                  $y_{j_t} \in[m_B]$, possibly stochastically, using $s,m'$ and
                  his previous choices and outcomes
                  $(j_u,y_{j_u},b_{j_u})_{u<t}$. He then obtains $b_{j_t}$.
            \item \emph{Decoding.} After all measurements, Bob produces a guess
                  $g$ using the stochastic map
                  $p_G(g|s,m',\mathbf j,\mathbf y,\mathbf b)$.
        \end{enumerate}
        In every realization of the protocol, the tuples $\boldsymbol{\ell}$ and
        $\mathbf j$ are permutations of $1,\ldots,c$. The vectors
        $\mathbf x,\mathbf a,\mathbf y, \mathbf b$ are ordered by copy label.
\end{dfn}

The previous copy choices and inputs are included explicitly because, for
    stochastic strategies, they cannot in general be reconstructed from the
    outputs alone.
 The measurement stage is analogous to the notion of
wiring of~\cite{Beigi2015}, specialized to identical copies and to effective
outputs that record all local outputs.
Allowing order adaptivity on both sides generalizes the usual
    concatenation procedures~\cite{pawlowski:2009}. 
The explicit joint distribution generated by a protocol is given in
Appendix~\ref{app:generalized_many_copy_protocols}.

For a fixed database distribution $q$ and a $c$-copy protocol $\pi_c$,
    we define
    \begin{equation}
        \begin{aligned}
            \Delta_c(P,q,\pi_c) & \coloneqq \sum_{s=0}^{n-1}I(R_s:G_s \,|\,\mathbf R_{<s}) \\
                                & \qquad - I(\mathbf R,\mathbf L,\mathbf X,\mathbf A:M').
        \end{aligned}
        \label{eq:generalized_many_copy_delta}
    \end{equation}
    This is the natural multi-copy extension of
    Eq.~\eqref{eq:generalized_single_copy_gap}, to which it reduces
    when $c = 1$. For adaptive orderings,
    $\mathbf L$ records the order in which Alice used the copies.
    Again, by data processing, factoring the same channel through an
    intermediate message can only
    weaken this bound. The corresponding optimized violation is
\begin{equation}
    \Delta_c^{*}(P) := \sup_{\substack{q\in\set{D}_{n,d}\\ \pi_c\ \mathrm{compatible}}} \Delta_c(P,q,\pi_c). \label{eq:generalized_many_copy_optimization}
\end{equation}
As in the single-copy problem, $q$ is a free variable of the optimization.

Note that Alice completes all her local operations before Bob starts his, so
Alice's choices never depend on Bob's outcomes, and Bob's choices depend on
Alice's data only through $M'$. Together with the fact that each \emph{decision
tree} uses every copy at most once in any realization, this ensures that the
adaptive protocol is always acyclic; details are given in
Appendix~\ref{app:generalized_many_copy_protocols}.

The deterministic reduction also extends to this setting. 

\begin{lemma}[Multi-copy deterministic reduction]
    \label{lem:generalized_deterministic_many_copy}
    Let $P$ be fixed, let $c < \infty$, and let $\set{D}_{n,d}$ be a fixed
    family of admissible database distributions. In the optimization of
    Eq.~\eqref{eq:generalized_many_copy_optimization}, Alice's and Bob's
    complete adaptive measurement decision trees, including the order in
    which each party uses the copies,
    as well as Bob's decoding map $T_G$, may be taken deterministic. In
    general, $q$ and $T_M$ remain unrestricted stochastic variables.
\end{lemma}

For fixed $q$, the IC deficit $\Delta_c$ is convex in each measurement or
    decoding map when all other maps are held fixed. As in the single-copy case,
    these maps can therefore be taken deterministic; see
    Appendix~\ref{app:generalized_deterministic_many_copy} for details and a
    proof.

After the reduction of
    Lemma~\ref{lem:generalized_deterministic_many_copy}, Bob's first copy and
    input are determined by $(S,M')$. Its output is an entry of $\mathbf B$,
    which determines his next choice, and so on. Hence $\mathbf J,\mathbf Y$
    are functions of $(S,M',\mathbf B)$ and
    can be substituted into the decoding map, which can then be written as
    $g = f_G(s,m',\mathbf b)$.
    Similarly, $\mathbf L,\mathbf X$ are functions of
    $(\mathbf R,\mathbf A)$, so
    \begin{equation}
        I(\mathbf R,\mathbf L,\mathbf X,\mathbf A:M') = I(\mathbf R,\mathbf A:M').
    \end{equation}

Thus, without loss of generality, we can restrict ourselves to such
deterministic protocols and consider simply the expression
\begin{equation}
        \begin{split}
            \Delta_c(P,q,\pi_c) & = \sum_{s=0}^{n-1}I(R_s:G_s\,|\,\mathbf R_{<s}) \\
                                & \qquad - I(\mathbf R,\mathbf A:M').
        \end{split}
        \label{eq:generalized_many_copy_delta_reduced}
    \end{equation}

We now connect this protocol to generalized wirings.
The observation we need is that a deterministic adaptive use of the $c$
copies can itself be regarded as a local wiring producing one effective
behavior. Define
\begin{equation}
    x' := \rr,\quad y'
       := (s,m'),\quad a'
       := \mathbf a,\quad b'
       := \mathbf b, \label{eq:generalized_effective_wiring_variables}
\end{equation}
where the outputs are again grouped according to the copy labels rather
than the step at which a copy is used. For fixed external inputs $(x',y')$ and
 outputs $(a',b')$, the deterministic decision trees recursively
determine all internal inputs. To see this, notice that on Alice's side,
$x' = \rr$ fixes the first action $(\ell_1,x_{\ell_1})$, the output
$a_{\ell_1}$ read from $a'$ then fixes the second action, and so on, so
that every (copy-labelled) input is a function
\begin{equation}
    x_i = x_i(\rr,\mathbf a),\qquad i=1,\ldots,c, \label{eq:alice_internal_inputs}
\end{equation}
which depends only on $\rr$ and the outputs observed by Alice before she
measures copy $i$.
Similarly, Bob's input $y_j = y_j(s,m',\mathbf b)$ depends only on $(s,m')$ and
    the outputs observed before he measures copy
    $j$. The effective behavior is then
\begin{equation}
    p'(a',b'|x',y') = \prod_{i=1}^c p(a_i,b_i|x_i,y_i). \label{eq:generalized_adaptive_effective_behavior}
\end{equation}
 This is a normalized nonsignaling behavior obtained from
$P^{\otimes c}$ by a local deterministic wiring, in which the order of use of
the copies may be different for the two parties and may depend, on each side, on
the local input and on the previous local outputs. Normalization and
nonsignaling follow from the results of Ref.~\cite{Beigi2015}; see
Lemma~\ref{lem:sequential_marginals} in
Appendix~\ref{app:generalized_effective_ns}. Its Bell scenario has $d^n$ inputs
and $d_A^c$ outputs for Alice, one input for every value of $\rr$, and
$nd'$ inputs and $d_B^c$ outputs for Bob, one input for every pair
$(s,m')$. The database distribution $q$ does not enter the definition of
$P'$: it only determines how the effective input $x' = \rr$ is sampled in the
resulting single-copy generalized IC protocol.

The communication channel also remains outside the wiring. In the effective
single-copy protocol Alice uses $x' = \rr$, obtains $a' = \mathbf a$, and
applies the channel
$p_M(m'|\rr,\mathbf a)$, obtained by substituting the deterministic
    functions $\boldsymbol\ell = \boldsymbol\ell(\rr,\mathbf a)$ and
    $\mathbf x = \mathbf x(\rr,\mathbf a)$ into the original channel
    $p_M(m'|\rr,\boldsymbol\ell,\mathbf x,\mathbf a)$.
Bob uses the effective input $y' = (s,m')$ and receives $b' = \mathbf b$.
He then applies the function
    $g = f_G(s,m',\mathbf b)$.
Consequently the adaptive and effective protocols generate exactly the same
joint distribution of $(\mathbf R,S,M',\mathbf A,\mathbf B,G)$ for every
fixed database distribution $q$.

\begin{obs}[Adaptive multi-copy protocols as generalized single-copy protocols]
    \label{obs:generalized_adaptive_many_copy_as_single_copy}
    For every fixed database distribution $q \in \set{D}_{n,d}$, any
     adaptive $c$-copy generalized IC protocol applied to $P$,
    with deterministic measurement choices, is equivalent to a
    generalized single-copy protocol applied to the effective behavior $P'$ of
    Eq.~\eqref{eq:generalized_adaptive_effective_behavior}. In particular, there
    exists a single-copy protocol $\pi_1$ such that
    \begin{equation}
        \Delta_c(P,q,\pi_c) = \Delta_1(P',q,\pi_1). \label{eq:generalized_many_to_one_gap_equality}
    \end{equation}
\end{obs}

 Allowing the database distribution to be a free parameter of the optimization does
not modify the wiring reduction: a database distribution witnessing a
multi-copy violation is carried unchanged to the corresponding effective
single-copy protocol.

\subsection{Single-copy sufficiency}
\label{sec:generalized_single_copy_sufficiency}

Observation~\ref{obs:generalized_adaptive_many_copy_as_single_copy} becomes
particularly useful for families of correlations that are stable under the
general wirings considered above. We remark that closure under arbitrary
wirings of different behaviors is not required. It is sufficient to assume
closure under local wirings of finitely many identical copies. More precisely,
let
\begin{equation}
    \set C = \{\set C(\mathsf B)\}_{\mathsf B}
\end{equation}
be a family of correlations indexed by the Bell scenario. We say that $\set C$
is closed under identical-copy wirings if, for every $P \in \set C(\mathsf B)$,
every finite $c$, and every local wiring $W$ of $c$ identical copies of $P$,
\begin{equation}
    W(P^{\otimes c}) \in \set C(\mathsf B'), \label{eq:generalized_identical_copy_closure}
\end{equation}
where $\mathsf B'$ denotes the Bell scenario of the effective behavior.

A particularly interesting case is that of the macroscopically local set $\set{Q}^1$
and the almost-quantum set $\set{AQ}$. Their closure under the adaptive wirings
discussed here was either already proven in the original
paper~\cite{navascues:2010} or follows directly from combining the results
in~\cite{navascues:2015}. For convenience, we
 give the argument in Appendix~\ref{app:closure_adaptive_wirings}.

\begin{prop}[Single-copy sufficiency for generalized IC]
    \label{prop:generalized_single_copy_sufficiency}
    Fix RAC parameters $(n,d,d')$ and a fixed family of admissible
    database distributions $\set{D}_{n,d}$. Let
    $\set C = \{\set C(\mathsf B)\}_{\mathsf B}$ be a family of correlations
    satisfying Eq.~\eqref{eq:generalized_identical_copy_closure}. Suppose that,
    for every Bell scenario $\mathsf B$ and every $P \in \set C(\mathsf B)$,
    single-copy IC is satisfied, i.e.
    \begin{equation}
        \Delta_1^{*}(P) \leq 0, \label{eq:generalized_single_copy_hypothesis}
    \end{equation}
    where the optimization in $\Delta_1^{*}$ includes also all
    $q \in \set{D}_{n,d}$. Then, for every finite $c$ and every
    $P \in \set C(\mathsf B)$,
    \begin{equation}
        \Delta_c^{*}(P) \leq 0. \label{eq:generalized_many_copy_conclusion}
    \end{equation}
    Thus generalized single-copy IC for all admissible
    database distributions is sufficient to exclude all finite-copy
    adaptive violations considered above.
\end{prop}

\begin{proof}
    Suppose, by contradiction, that some $P \in \set C(\mathsf B)$ violates
    generalized IC through a finite-copy adaptive protocol. Then there exist a
    database distribution $q \in \set{D}_{n,d}$, a finite $c$, and a
    compatible protocol $\pi_c$ such that
    \begin{equation}
        \Delta_c(P,q,\pi_c) > 0. \label{eq:generalized_many_copy_positive_witness}
    \end{equation}
    By Lemma~\ref{lem:generalized_deterministic_many_copy}, the local adaptive
    operations may be taken deterministic without decreasing the violation,
    while keeping the same database distribution $q$.
    Observation~\ref{obs:generalized_adaptive_many_copy_as_single_copy} then
    gives an effective behavior
    \begin{equation}
        P' = W(P^{\otimes c})
    \end{equation}
    and a generalized single-copy protocol $\pi_1$ such that, for that same
    database distribution $q$,
    \begin{equation}
        \Delta_1(P',q,\pi_1)
        = \Delta_c(P,q,\pi_c)
        > 0. \label{eq:generalized_single_copy_positive_witness}
    \end{equation}
    Since $P'$ is obtained from identical copies of $P$ by a local wiring,
    Eq.~\eqref{eq:generalized_identical_copy_closure} implies
    $P' \in \set C(\mathsf B')$. This contradicts the assumption of the validity
    of Eq.~\eqref{eq:generalized_single_copy_hypothesis} for all
    $P' \in \set C(\mathsf B')$ and all $\mathsf B'$.
\end{proof}

Thus, for a family closed under these wirings, satisfying generalized IC for all
single-copy protocols is equivalent to satisfying it for all finite-copy
protocols.

\subsection{Sets of correlations satisfying generalized information causality}
\label{sec:generalized_ic_sets}

We now use the previous characterization to define the sets of correlations
satisfying generalized IC.
Let us consider a set of allowed database distributions $\set{D}_{n,d}$.
    Restricting the distributions to any subset
    $\widetilde{\set{D}}_{n,d} \subseteq \set{D}_{n,d}$ gives
    \begin{equation}
        \sup_{q \in \widetilde{\set{D}}_{n,d}, \, \pi_c} \Delta_c(P,q,\pi_c) \leq \sup_{q\in {\set{D}}_{n,d}, \, \pi_c} \Delta_c(P,q,\pi_c). \label{eq:generalized_source_family_monotonicity}
    \end{equation}
    We therefore take $\set{D}_{n,d}$ to be the full probability simplex on
    $[d]^n$, which gives the strongest IC constraint. Indeed, it requires the IC
    inequality to hold for all possible choices of $q(\rr)$, and this is the
    most
    restrictive requirement.
    Consequently, in this subsection, we omit the argument $\set{D}_{n,d}$ from
    the
    notation.

For a fixed Bell scenario $\mathsf B$, RAC parameters $(n,d,d')$, and a fixed
number $c$ of copies, let us denote the quantity from
Eq.~\eqref{eq:generalized_many_copy_optimization} with these explicit labels
\begin{equation}
    \begin{split}
        \Delta_{n,d,d'}^{*,(c)} (P) \coloneqq \sup_{\substack{q\in\set{D}_{n,d}\\
\pi_c\ \mathrm{compatible}}} \Delta_c(P,q,\pi_c).
    \end{split}
    \label{eq:generalized_optimal_c_copy_violation}
\end{equation}
As the objective function is continuous and the domain compact, the
    supremum is attained.
We define
\begin{equation}
    \set{IC}_{n,d,d'}^{(c)} (\mathsf B) \coloneqq \bigl\{P \in \set{NS}(\mathsf B)\ \bigm|\  \Delta_{n,d,d'}^{*,(c)} (P) \leq 0\bigr\}. \label{eq:generalized_fixed_c_ic_set}
\end{equation}
Thus membership requires the generalized IC inequality to hold for every
 $q$ and every compatible protocol with the specified finite alphabets.

Intersecting over the RAC parameters gives a $c$-copy information
    causality set,
\begin{equation}
    \set{IC}^{(c)}(\mathsf B) \coloneqq \bigcap_{n,d,d' \geq 2} \set{IC}_{n,d,d'}^{(c)} (\mathsf B). \label{eq:generalized_c_copy_ic_set}
\end{equation}
In particular,
$\set{IC}^{(1)}(\mathsf B)$ contains the behaviors satisfying every generalized
single-copy IC constraint. We also define 
\begin{equation}
    \set{IC}^{(\infty)}(\mathsf B) \coloneqq \bigcap_{c \geq 1}\set{IC}^{(c)}(\mathsf B). \label{eq:generalized_full_ic_set}
\end{equation}
By construction,
$\set{IC}^{(\infty)}(\mathsf B) \subseteq \set{IC}^{(1)}(\mathsf B)$. Combined
with Proposition~\ref{prop:generalized_single_copy_sufficiency}, this gives
\begin{prop}[Single-copy sufficiency for closed sets]
For every family $\set C$ closed under identical-copy wirings,
\begin{equation}
    \begin{aligned}
        \set C(\mathsf B) \subseteq \set{IC}^{(1)}      & (\mathsf B) \quad\text{for every }\mathsf B  \\[-0.2em]
                                                        & \Longleftrightarrow                          \\[-0.2em]
        \set C(\mathsf B) \subseteq \set{IC}^{(\infty)} & (\mathsf B) \quad\text{for every }\mathsf B.
    \end{aligned}
    \label{eq:generalized_self_consistent_single_copy_equivalence}
\end{equation}
\end{prop}
Finally, we can show the convexity of all the single-copy IC sets.

\begin{prop}[Convexity of single-copy set]
    \label{prop:generalized_convexity_single_copy_ic}
    For a fixed Bell scenario $\mathsf B$ and fixed RAC parameters
    the function
    \begin{equation}
        P \mapsto \Delta_{n,d,d'}^{*,(1)} (P)
    \end{equation}
    is convex. Consequently, $\set{IC}_{n,d,d'}^{(1)} (\mathsf B)$ and
    $\set{IC}^{(1)}(\mathsf B)$ are convex.
\end{prop}
The proof is given in Appendix~\ref{app:generalized_single_copy_conv}.
For fixed $q$, define
\begin{equation}
    \Delta_{n,d,d'}^{*,(1)}(P;q) \coloneqq \sup_{\pi_1}\Delta_1(P,q,\pi_1). \label{eq:generalized_fixed_source_optimum}
\end{equation}
 This is a convex function of $P$. The
quantity optimized over protocols and database distributions is simply the pointwise supremum
\begin{equation}
    \Delta_{n,d,d'}^{*,(1)} (P) = \sup_{q \in \set{D}_{n,d}} \Delta_{n,d,d'}^{*,(1)}(P;q), \label{eq:generalized_source_pointwise_supremum}
\end{equation}
which remains convex.
In particular, the database distribution $q$ maximizing
Eq.~\eqref{eq:generalized_source_pointwise_supremum} may depend on $P$, as it
does in the correlated-input construction of Ref.~\cite{gachechiladze:2026},
without affecting convexity of the optimized value.

We note that the same proof does not directly establish convexity for $c > 1$
as, for a fixed adaptive $c$-copy protocol, the joint probabilities depend on
the shared behavior through $P^{\otimes c}$ and are therefore not affine in $P$.

\section{Lower and upper bounds on IC violations}
\label{sect:generalized_opti}

The results of Secs.~\ref{sec:generalized_single_copy_protocols}
and~\ref{sec:generalized_many_copy_ic} reduce the comparison between generalized
IC and a set of behaviors
$\set C = \set C(\mathsf B) \subseteq \set{NS}(\mathsf B)$ to a single-copy
optimization problem.
We therefore formulate the bounds for a single copy, to also simplify the
    presentation.
For fixed RAC parameters $(n,d,d')$ and 
fixed $\set D_{n,d}$, let
\begin{equation}
    \Delta^*_{\set C} \coloneqq \sup_{P \in \set C(\mathsf B)} \Delta_{n,d,d'}^{*,(1)}(P;\set D_{n,d}) \label{eq:proposed_generalized_set_level_ic_optimum}
\end{equation}
be the largest violation attainable within $\set C$. By definition, $\Delta^*_{\set C} \leq 0$ if and only if
$\set C(\mathsf B) \subseteq \set{IC}_{n,d,d'}^{(1)}(\mathsf B;\set D_{n,d})$.

After the deterministic reduction of Lemma~\ref{lemma:generalized_det_strat},
\begin{equation}
    \Delta^*_{\set C} = \max_{\mathbf f \in \set F_{}} \sup_{\substack{P\in\set C,\ q\in\set D_{n,d}\\ T_M\in\set T_M}} \Delta_1(P,q,\mathbf f,T_M), \label{eq:proposed_generalized_set_level_reduced}
\end{equation}
where $\set F_{}$ is the finite set of deterministic triples
$\mathbf f = (f_X,f_Y,f_G)$,
and $\set T_M$ is the set of stochastic communication matrices. The problem is
nonconvex and not easy to solve. We can therefore bound the optimum from below and above:
\begin{equation}
    L_{\set C}
    \leq \Delta^*_{\set C}
    \leq U_{\set C}. \label{eq:proposed_generalized_two_sided_ic_bounds}
\end{equation}
A lower bound $L_{\set C} > 0$ from an explicit feasible point exhibits a
violation within $\set C$, while an upper bound $U_{\set C} \leq 0$ from a
relaxation proves the inclusion.

Every feasible point of Eq.~\eqref{eq:proposed_generalized_set_level_reduced}
gives a lower bound on $\Delta^*_{\set C}$, at which $\Delta_1$ can be evaluated
directly, and a positive value
certifies a violation, with an explicit witness \mbox{$(P,q,f_X,f_Y,f_G,T_M)$}.
For each \mbox{$\mathbf f$}, up to relabeling symmetries, the continuous
variables can be optimized by standard local methods, such as gradient ascent
\mbox{\cite{kingma:2017}}, from several initial points.
One possible route to upper bounds is to express $\Delta_1$ as a linear
    combination of joint entropies and approximate the entropy functions by
    polynomials. For example, one can approximate $\eta(p) = -p\log_2p$
    uniformly on $[0,1]$ with error at most $\epsilon$. If the entropy expansion
    contains $N$ occurrences of $\eta(p)$, the total error in the IC deficit is at
    most $N\epsilon$. Maximizing the polynomial approximation of $\Delta_1$ over
    the feasible set and adding $N\epsilon$ gives an upper bound on
    $\Delta_{\mathcal C}^*$.

If $\set C(\mathsf B)$ is polyhedral, one may optimize directly over the
behavior probabilities $p(a,b|x,y)$ together with the linear normalization,
nonsignaling, and positivity constraints defining the set. Moment/SOS
relaxations~\cite{lasserre:2001} then provide upper bounds. If $\set C$ is
described by a positive semidefinite moment matrix, as for
 macroscopic locality~\cite{navascues:2010, cavalcanti:2010},
 almost-quantum correlations~\cite{navascues:2015}, or
higher NPA levels~\cite{navascues:2007, navascues:2008}, the products of
behavior probabilities in the
polynomial upper bound
can be treated with scalar-extension/state-polynomial relaxations~\cite{pozas:2019,
klep:2022, klep:2024}; $q$ and $T_M$ enter as additional commuting
variables.
The special case
$\set C = \{P\}$
    gives the membership problem for the single-copy IC set.
By the convexity of $P \mapsto \Delta^*_{\{P\}}$
(Proposition~\ref{prop:generalized_convexity_single_copy_ic}), the
single-behavior  analysis also applies to polytopes: if $\set C$ is the convex hull of
$P_1,\ldots,P_k$, then $\Delta^*_{\set C} = \max_i\Delta^*_{\{P_i\}}$. 
For a general compact convex $\set C$, the supremum is attained at an extreme 
point, since $\Delta^{*,(1)}_{n,d,d'}$ is convex and, for fixed RAC parameters, 
also continuous in $P$. Searches can therefore be restricted to extreme points.

The general optimization strategy outlined here remains computationally
    demanding. In Sec.~\ref{sec:low_capacity}, we therefore focus on
    channels close to a completely noisy one.  For perturbations that preserve the full support of the channel output, we derive the quadratic term of the IC functional.
    This gives an analytically tractable application of our protocol characterizations and 
    allows us to characterize all second-order constraints in the simplest Bell scenario. 
    For perturbations that change the channel support, we show instead that the leading term 
    is linear and reproduces the full single-copy IC constraint.

\section{Applications: expansions around zero-capacity channels}
\label{sec:low_capacity}

Recently Gachechiladze and Miklin \cite{gachechiladze:2026} have shown that IC implies
the Tsirelson--Landau--Masanes inequality \cite{tsirelson:1987, landau:1988, masanes:2003} for renormalized correlators~\cite{navascues:2007}
 in the CHSH, i.e., $(2,2,2,2)$, scenario. Their proof
was based on a vanishing-channel-capacity expansion of the IC expression, based on the 
technique developed in \cite{jain:2024}, together with a smart
choice of the protocol ($q$,$\pi$). 
In the same paper, the authors ask the 
question of whether a different protocol may yield constraints
beyond ML already at leading order in the vanishing-capacity limit.
In this section, we apply our characterization of IC as an optimization problem to 
answer it in the following way. Using the method of Ref.~\cite{jain:2024}, namely
expanding around the maximally noisy zero-capacity channel that produces a uniform output, and more generally
around any zero-capacity channel with full support, does not yield any improvement at leading
order over the TLM boundary derived in Ref.~\cite{gachechiladze:2026}. On the contrary, 
using as the center of the expansion a degenerate channel (i.e., one for which some 
output symbols have probability $0$), and considering perturbations that do not
preserve the support of the channel, one can obtain any single-copy IC violation at linear order 
in the perturbation. In other words, any such violation can be detected arbitrarily close to 
a zero-capacity channel.

\subsection{Second-order expansion}
The first observation is that generalized IC must hold for every channel, including channels arbitrarily close
to a completely noisy one, i.e., one whose output does not depend on its input. Here, 
we expand the IC functional around such channels whose output distribution has
full support. The leading term is a
quadratic form in the channel perturbation, which we derive for every single-copy protocol;
Sec.~\ref{sec:second_order} evaluates the resulting constraints in the
$(2,2,2,2)$ scenario. In the following, we take as $G_s$ the full information
available to Bob after the deterministic reduction, namely $G_s:=(M',B_s)$. Indeed, the 
preceding results hold for any finite alphabet of $G$, and by the (conditional) data-processing inequality, this choice provides the strictest condition. 
We use natural logarithms throughout this section.
For convenience, we write the IC functional of
Eq.~\eqref{eq:generalized_opt_prob_simplified} with the 
opposite sign, as an IC margin:
\begin{equation}
\begin{aligned}
    \Lambda \coloneqq -\Delta_1 &=  I(\mathbf R,A:M') \\ & \qquad \qquad -\sum_{s=0}^{n-1}I(R_s:M',B_s\,|\,\mathbf R_{<s}),
\label{eq:margin_def}
\end{aligned}
\end{equation}
so that IC requires $\Lambda\geq0$ for every protocol.

We consider protocols after the deterministic reduction of
Lemma~\ref{lemma:generalized_det_strat}: a database distribution $q$, Alice's setting
$x_\rr=f_X(\rr)$, Bob's setting $y_{m'}^{s}=f_Y(s,m')$, which may depend on the message,
and a channel $T_M$ acting on $(\rr,a)$. A completely noisy channel,
$[T_0]_{m',(\rr,a)}=\mu(m')$, makes $M'$ independent of the input $(\mathbf R, A)$, i.e.,
a zero-capacity channel, and gives $\Lambda=0$. 
Here, we take $\mu(m')>0$ for all $m'$ and write every channel in a
neighborhood of $T_0$ as
\begin{equation}
    \begin{gathered}
        [T_M]_{m',(\rr,a)} = \mu(m')\bigl[1 + \varepsilon\,w_{\rr,a}(m')\bigr],        \\
        \ \text{ with } \ \sum_{m'}\mu(m') = 1, \, \, \sum_{m'}\mu(m')\,w_{\rr,a}(m') = 0,
    \end{gathered}
    \label{eq:general_channel}
\end{equation}
where $w$ specifies the direction of the perturbation and $\varepsilon\geq0$ its size. The
constraint guarantees normalization, and positivity holds for $\varepsilon$ small enough.
This expansion generalizes those in~\cite{jain:2024, gachechiladze:2026} that are centered
in the maximally noisy channel $\mu(m')=1/d'$.

Since all terms of Eq.~\eqref{eq:margin_def} are mutual informations between a
variable and the message, possibly together with Bob's outcome, it is useful
to write down the following expansion of mutual information for 
\emph{very noisy} channels~\cite{abbe:2010}, close to zero capacity.

\begin{lemma}[Mutual information near zero capacity]
\label{lemma:mi_near_zero}
Let $p(z|v)=\mu(z)\bigl[1+\varepsilon\,u_v(z)\bigr]$, with a distribution $\mu$ that does
not depend on $v$ and $\sum_z\mu(z)u_v(z)=0$ for every $v$. Then
\begin{equation}
I(V:Z)=\frac{\varepsilon^2}2 \sum_z\mu(z)\Var_V\bigl(u_V(z)\bigr)+O(\varepsilon^3)
\label{eq:mi_second_order}
\end{equation}
where the variance is taken over $V$.
\end{lemma}

\begin{proof}
For any such $\mu$, 
\begin{equation}
\begin{aligned}
    I(V:Z) & = \sum_{v,z} p(v)p(z|v) \ln\frac{p(z|v)}{p(z)}                                                    \\
           & = \sum_{v,z} p(v)p(z|v) \biggl[\ln\frac{p(z|v)}{\mu(z)} - \ln\frac{p(z)}{\mu(z)}\biggr] \\
           & = \sum_v p(v)\,D\bigl(p_Z(\cdot | v)\,\|\,\mu\bigr) - D(p_Z\,\|\,\mu),
\end{aligned}
\end{equation}
where $D(\cdot \| \cdot)$ is the classical relative entropy and we used 
$\ln({p(z|v)}/p(z))=\ln({p(z|v)}/{\mu(z)})+\ln({\mu(z)}/{p(z)})$. The marginal
$p_Z$ is given by averaging over $v$ the conditional distribution $p(z|v)$, 
hence, $p_Z=\mu\,(1+ \varepsilon\,\EE_V[u_V])$.
Using $(1+t)\ln(1+t)=t+\tfrac12t^2+O(t^3)$ and $\sum_z\mu(z)u(z)=0$, we get for the two terms
\begin{equation}
\begin{aligned}
    & D\bigl(\mu(1 + \varepsilon u)\|\mu\bigr) = \frac{\varepsilon^2}{2}\sum_z\mu(z)u(z)^2 + O(\varepsilon^3), \\[3pt]
    & D\bigl(\mu(1 + \varepsilon\mathbb E_V[u_V])\|\mu\bigr)                                                   \\
    & \qquad \ \qquad = \frac{\varepsilon^2}{2} \sum_z\mu(z)\mathbb E_V[u_V(z)]^2 + O(\varepsilon^3).
\end{aligned}
\end{equation}
Finally, the difference of the two terms gives precisely the average variance term appearing in 
Eq.~\eqref{eq:mi_second_order}.
\end{proof}

For Alice's term, $V=(\mathbf R,A)$, $Z=M'$, and the perturbation is the channel
direction $w_{\rr,a}$ itself. For Bob's terms when $S=s$, the relevant
variable is his full information $Z=(M',B_s)$. Summing over $a$ and using
nonsignaling, $p(a,b|x,y)=p(b|y)p(a|b,x,y)$, its distribution given $\rr$ is
\begin{equation}
\begin{split}
p(m',b|\rr)&=\sum_a p(b|y^s_{m'})p(a|b,x_\rr,y^s_{m'}) p(m'|\rr, a)\\
&=\mu(m')\,p(b|y^s_{m'})\bigl[1+\varepsilon\,\zeta^{(s)}_\rr(m',b)\bigr],\\
\end{split}
\label{eq:bob_transcript}
\end{equation}
where $y^s_{m'}=f_Y(s,m')$ and
\beq
\zeta^{(s)}_\rr(m',b) :=\sum_ap(a|b,x_\rr,y^s_{m'})\,w_{\rr,a}(m').
\eeq
The distribution $\mu(m')p(b|y^s_{m'})$ does not
depend on $\rr$, and $\zeta^{(s)}_\rr(m',b)$ can be seen as the posterior average of the
channel direction given Bob's outcome: Bob sees Alice's data $(\rr, a)$ only through what
$b$ reveals about $a$.

To obtain the second-order form of the IC functional, we apply Lemma~\ref{lemma:mi_near_zero} to Alice's
term, and to Bob's terms written as $I(R_s:M',B_s \mid \mathbf R_{<s})=I(\mathbf R_{\leq s}:M',B_s)-I(\mathbf R_{<s}:M',B_s)$. We, then, obtain
\begin{prop}[Second-order form]
\label{prop:second_order_form}
For every single-copy protocol and every direction $w$,
\begin{equation}
\begin{gathered}
\Lambda=\frac{\varepsilon^2}2 \,\Phi_2(w)+O(\varepsilon^3),\\
\Phi_2(w)=\sum_{m'}\mu(m')\,w(m')^{\mathsf T}K_{f_X,f_Y^{m'}}(P,q)\,w(m'),
\end{gathered}
\label{eq:second_order_form}
\end{equation}
where $w(m')=(w_{\rr,a}(m'))_{\rr,a}$, and $f_Y^{m'}(s):=f_Y(s,m')$ is the
fixed-message strategy and the symmetric matrix $K_f(P,q)$, for a fixed-message
strategy $f=(f_X,f_Y)$ with $y_s=f_Y(s)$ and $\zeta^{(s)}_\rr(b)=\sum_a
p(a|b,x_\rr,y_s)\,w_{\rr,a}$,  is defined as
\begin{equation}
\begin{split}
&w^{\mathsf T}K_f\,w \coloneqq \Var(w_{\mathbf R,A})
-\sum_{s=0}^{n-1}\sum_bp(b|y_s)\\
&\ \ \times\Bigl[\Var\bigl(\EE[\zeta^{(s)}_{\mathbf R}(b)|\mathbf R_{\leq s}]\bigr)-\Var\bigl(\EE[\zeta^{(s)}_{\mathbf R}(b)|\mathbf R_{<s}]\bigr)\Bigr].
\end{split}
\label{eq:hessian_matrix}
\end{equation}
Here, the joint distribution of $(\mathbf R, A)$ is $q(\rr)p(a|x_\rr)$.
\end{prop}
A detailed proof is provided in Appendix~\ref{app:second_order_form}. 
Notice that for the special case of a binary message with uniform 
$\mu$ and message-independent $f_Y$, the directions satisfy $w(1)=-w(0)$, and 
$\Phi_2=w(0)^{\mathsf T}K_fw(0)$. In this
case $K_f$ is, up to a multiplicative factor, the Hessian of $\Lambda$, i.e., the matrix of
second derivatives with respect to $T_M$, evaluated at the zero-capacity channel, as one would have expected 
from a second-order expansion. 
On the other hand, the effect of the channel distribution $\mu$ is simply to give a convex 
mixture of matrices $K_{f_X,f_Y^{m'}}$, each corresponding to a fixed message-independent
strategy. We now show that for this reason it is enough to restrict ourselves to 
message-independent settings $f_Y$ and binary messages.

First, it is helpful to define the set characterized by this second-order approximation as
\begin{equation}
\set{IC}^\mathrm{2nd}(\mathsf B) := \left\{P\in\set{NS}(\mathsf B) \middle|
\Phi_2(w) \geq 0, \forall n,d,d', q,f,\mu,w \right\},
\end{equation}
namely, $\set{IC}^\mathrm{2nd}(\mathsf B)$ contains the behaviors that satisfy  $\Phi_2(w) \geq 0$ for all single-copy protocols, RAC parameters $(n,d,d')$, full-support noisy channels $\mu$, and directions $w$
in Eq.~\eqref{eq:general_channel}.

If $\Phi_2(w)<0$, Proposition~\ref{prop:second_order_form} gives $\Lambda<0$ for
all sufficiently small $\varepsilon>0$, and the protocol with channel
\eqref{eq:general_channel} gives an explicit violation. Hence
$\set{IC}^{(1)}(\mathsf B)\subseteq\set{IC}^\mathrm{2nd}(\mathsf B)$.

\begin{coro}
\label{coro:ic2_characterization}
$P\in\set{IC}^\mathrm{2nd}(\mathsf B)$ if and only if $K_f(P,q)\succeq0$ for all RAC
parameters $(n,d)$, all database distributions $q$, all $f_X$ and all message-independent $f_Y$.
\end{coro}

\begin{proof}
If all forms $K_f$ are positive semidefinite, every term of
Eq.~\eqref{eq:second_order_form} is nonnegative. Conversely, if
$v^{\mathsf T}K_fv<0$, the protocol with a binary message, uniform $\mu$,
message-independent $f_Y$ and $w(0)=-w(1)=v$ gives $\Phi_2=v^{\mathsf T}K_fv<0$.
\end{proof}

We conclude that message-dependent settings and larger message alphabets give no
additional second-order constraints, as it suffices
to consider binary messages, uniform $\mu$, and message-independent Bob settings 
when testing membership in $\set{IC}^\mathrm{2nd}$. The above argument also provides 
an intuition for the emergence of positive semidefinite constraints in the 
L'Hôpital-rule derivation by Gachechiladze and Miklin \cite{gachechiladze:2026}. 
Indeed, the condition $\Lambda \geq 0$ together with the vanishing of $\Lambda$ and its 
first derivative for the completely noisy channel implies that the Hessian must be positive 
semidefinite, otherwise one could find a direction in the channel parameters where $\Lambda<0$.

\subsection{Application to the CHSH scenario}
\label{sec:second_order}

We now evaluate the second-order constraints in the CHSH scenario $\mathsf B = (2,2,2,2)$, 
and show that they characterize exactly the macroscopically local set $\set{Q}^1$.
Since second-order constraints are necessary for generalized IC, it follows that
in this scenario every behavior satisfying single-copy generalized IC 
is macroscopically local.

From now on the outcomes are binary, $a,b\in\{0,1\}$ and we use the correlator
parametrization of the behavior~\cite{brunner:2014}:
\begin{equation}
\begin{split}
p(a,b|x,y)=\frac 1 4 \bigl[1&+(-1)^a\langle a_x\rangle+(-1)^b\langle b_y\rangle\\
&+(-1)^{a\oplus b}\langle a_xb_y\rangle\bigr],\\
\end{split}
\label{eq:biased_behavior}
\end{equation}
where
\beq
\begin{gathered}
    \langle a_x\rangle := \sum_{a}(-1)^a p(a|x), \quad \langle b_y\rangle := \sum_{b} (-1)^b p(b|y), \\
    \langle a_xb_y\rangle := \sum_{a,b}(-1)^{a+b}p(a,b|x,y).
\end{gathered}
\eeq
We  assume nondegenerate marginals, i.e., $|\langle a_x\rangle|<1$ and
$|\langle b_y\rangle|<1$; the degenerate case is trivial and treated explicitly 
in the proof of Theorem~\ref{thm:second_order_is_ML}. The relevant 
coordinates are given by the centered and normalized outcomes
\begin{equation}
\xi_x(a)=\frac{(-1)^a-\langle a_x\rangle}{\sqrt{1-\langle a_x\rangle^2}},
\qquad
\eta_y(b)=\frac{(-1)^b-\langle b_y\rangle}{\sqrt{1-\langle b_y\rangle^2}},
\label{eq:centered_outcomes}
\end{equation}
which satisfy $\EE[\xi_x]=\EE[\eta_y]=0$ and
$\EE[\xi_x^2]=\EE[\eta_y^2]=1$ (where expectation values are taken over outcomes for fixed settings), 
and by the renormalized correlators
\begin{equation}
D_{x,y}:=\EE\bigl[\xi_x \,\eta_y \bigr]
=\frac{\langle a_xb_y\rangle-\langle a_x\rangle\langle b_y\rangle}
{\sqrt{(1-\langle a_x\rangle^2)(1-\langle b_y\rangle^2)}},
\label{eq:renormalized_correlators}
\end{equation}
which satisfy $|D_{x,y}|\leq1$, since for normalized $\xi_x$ and $\eta_y$ we have
\begin{equation}
\begin{split}
0\leq\EE\bigl[(\xi_x-\lambda\,\eta_y)^2\bigr]=1-2\lambda D_{x,y}+\lambda^2,
\end{split}
\label{eq:D_bound}
\end{equation}
with the right-hand side attaining the minimum for $\lambda=D_{x,y}$.

These are the variables in which macroscopic locality is characterized in this
scenario~\cite{navascues:2007,navascues:2010}. 
Since $\{1,\xi_x\}$ is a basis of the functions of
a binary outcome, every channel can be written uniquely as
\begin{equation}
w_{\rr,a}=\alpha_\rr+\beta_\rr\,\xi_{x_\rr}(a),
\qquad x_\rr=f_X(\rr),
\label{eq:alpha_beta_split}
\end{equation}
where $\alpha_\rr$ is the outcome-independent part. 

Bob accesses Alice's outcome only through his own outcome $b$. For fixed
inputs $(x,y)$, the conditional mean $\EE[\xi_x(A)|B=b]$ is a function of
the binary variable $b$ with zero mean, since $\EE[\xi_x]=0$. Every function
of $b$ is a linear combination of $1$ and $\eta_y(b)$, and the zero-mean condition 
implies that the constant term vanishes, so that $\EE[\xi_x(A)|B=b]=\lambda\,\eta_y(b)$. 
Multiplying by $\eta_y(b)$ and averaging over $b$, with $\EE[\eta_y^2]=1$, gives
$\lambda=\EE[\xi_x(A)\eta_y(B)]=D_{x,y}$, i.e.,
\begin{equation}
\EE[\xi_x(A)|B=b]=D_{x,y}\,\eta_y(b).
\label{eq:xi_given_b}
\end{equation}
Inserting Eq.~\eqref{eq:alpha_beta_split} into the definition of
$\zeta^{(s)}_\rr(b)=\EE[w_{\rr,A}|\rr,b]$ of Proposition~\ref{prop:second_order_form},
with Bob's setting $y_s$, then gives
\begin{equation}
\zeta^{(s)}_\rr(b)=\alpha_\rr+\beta_\rr\,D_{x_\rr,y_s}\,\eta_{y_s}(b).
\label{eq:zeta_centered}
\end{equation}

We now have all the elements to compute explicitly $w^{\mathsf T}K_fw$.
\begin{lemma}[Second-order form for binary outcomes]
\label{lemma:second_order_structure}
Let $P$ be a nonsignaling behavior with binary outcomes
and nondegenerate marginals.
For a message-independent strategy $f$, in the coordinates
$(\alpha,\beta)$ the form $w^{\mathsf T}K_fw$ does not depend on $\alpha$, and
\begin{equation}
\begin{split}
w^{\mathsf T}K_fw=\sum_\rr q(\rr)\,\beta_\rr^2-\sum_{s=0}^{n-1}
\Bigl\{&\EE\bigl[\EE[g^{(s)}_{\mathbf R}|\mathbf R_{\leq s}]^2\bigr]\\
&-\EE\bigl[\EE[g^{(s)}_{\mathbf R}|\mathbf R_{<s}]^2\bigr]\Bigr\},
\end{split}
\label{eq:phi2_beta}
\end{equation}
where $g^{(s)}_\rr:=\beta_\rr\,D_{x_\rr,f_Y(s)}$ and expectations are taken over
$\mathbf R$, whose distribution is $q(\rr)$. In particular, this form
depends on $P$ only through the renormalized correlators $D_{x,y}$.
\end{lemma}

The explicit proof is given in Appendix~\ref{app:centered_outcomes}. The main
idea is the following. The terms mixing $\alpha$ and $\beta$ vanish 
because $\xi$ and $\eta$ have zero mean. The $\alpha$ part does not 
depend on $a$, so its contribution is independent of Bob's outcome $b$, 
and his terms reduce to the increments of $\Var(\EE[\alpha_{\mathbf R}|\mathbf R_{\leq s}])$
along $s$, which add up to Alice's term $\Var(\alpha_{\mathbf R})$. The
marginals enter only through the normalization of $\xi$ and $\eta$. In
Eq.~\eqref{eq:phi2_beta}, the first term is the variance of the channel
direction due to Alice's outcome, and the $s$-th term in braces is 
the difference of the variances for Bob's estimate of 
$g^{(s)}_\rr$ before and after conditioning on $R_s$.

Let us now consider the simplest case $n=d=2$. We write 
$\rr=(i,j)$ and parametrize the database distribution as $q(i,j)=p_i\tau_i(j)$, with 
$p_0=p$, $p_1=1-p$, $\tau_i(0)=t_i$ and $\tau_i(1)=1-t_i$. 
The rows of the truth table of $f_X$ are the maps $j\mapsto f_X(i,j)$; 
row $i$ collects the strings with $r_0=i$.

\begin{prop}[Second order in $(2,2,2,2)$]
\label{prop:second_order_2222}
Let $P$ have nondegenerate marginals. Then $K_f(P,q)\succeq0$ for all database distributions
and all strategies with $n=d=2$ if and only if $L(t)\succeq0$ for all
$t\in[0,1]$, where
\begin{equation}
L(t)=
\begin{pmatrix}
t\,\ell_0(t) & t(1-t)\,\kappa\\
t(1-t)\,\kappa & (1-t)\,\ell_1(t)
\end{pmatrix},
\label{eq:L_matrix}
\end{equation}
with
\begin{equation}
\begin{gathered}
\ell_0(t)=1-(1-t)D_{0,1}^2-t\,D_{0,0}^2,\\
\ell_1(t)=1-t\,D_{1,1}^2-(1-t)D_{1,0}^2,\\
\kappa=D_{0,1}D_{1,1}-D_{0,0}D_{1,0}.
\end{gathered}
\label{eq:L_entries}
\end{equation}
\end{prop}

A detailed proof is provided in Appendix~\ref{app:reduction_2222}. The main idea is 
simply to optimize over all possible strategies for Alice and Bob and all database distributions. 
Computing explicitly the expression in Eq.~\eqref{eq:phi2_beta}, the resulting
expression is $L(t)$, which only depends the database distribution parameter $t$. 
The final step is precisely to eliminate the parameter $t$ from the positivity 
condition $L(t)\succeq0$.

\begin{prop}[TLM inequality] $K_f(P,q)\succeq0$ for all database distributions and all strategies with $n=d=2$ if and only if
\begin{equation}
\begin{split}
\bigl|D_{0,0}D_{1,1}-D_{0,1}D_{1,0}\bigr|
\leq{}&\sqrt{(1-D_{0,1}^2)(1-D_{1,0}^2)}\\
&+\sqrt{(1-D_{0,0}^2)(1-D_{1,1}^2)},
\end{split}
\label{eq:tlm_landau}
\end{equation}
namely, the Tsirelson--Landau--Masanes (TLM) criterion in the form of
Landau~\cite{landau:1988}, written in the renormalized correlators. 
\label{prop:tlm_landau}
\end{prop}

We recall that for zero marginals $D_{x,y}=\langle a_xb_y\rangle$, 
and Eq.~\eqref{eq:tlm_landau} characterizes the quantum 
correlator region~\cite{landau:1988,masanes:2003}.
A detailed proof of Proposition \ref{prop:tlm_landau} is presented 
in Appendix \ref{app:proof_tlm}. We provide a sketch in the following. 
For $0<t<1$, we define $\widehat L(t)$ from the condition $L(t)=D_t\widehat L(t)D_t$ 
with $D_t=\operatorname{diag}(\sqrt t,\sqrt{1-t})$, explicitly
\begin{equation}
\widehat L(t) :=
\begin{pmatrix}
\ell_0(t) & \sqrt{t(1-t)}\,\kappa\\
\sqrt{t(1-t)}\,\kappa & \ell_1(t)
\end{pmatrix}.
\label{eq:hat_L}
\end{equation}
By construction $L(t)\succeq 0 \Leftrightarrow \widehat L(t) \succeq 0$.
One can reduce this condition to the equivalent one $\det\widehat L(t)\geq0$. 
Introducing the variables
\begin{equation}
\begin{gathered}
R:=(1-D_{0,1}^2)(1-D_{1,0}^2), \\
S:=(1-D_{0,0}^2)(1-D_{1,1}^2),\\
\Theta:=D_{0,0}D_{1,1}-D_{0,1}D_{1,0},
\end{gathered}
\label{eq:RS_Theta}
\end{equation}
one finds
\begin{equation}
\det\widehat L(t)=(1-t)\,R+t\,S-t(1-t)\,\Theta^2 .
\label{eq:det_hat_L}
\end{equation}
Equivalently, the positivity of $\widehat L(t)$ amounts to  $\Theta^2\leq R/t+S/(1-t)$.
The last step is to show that $\inf_t R/t+S/(1-t) = (\sqrt{R}+\sqrt{S})^2$
and substitute the values of $R$, $S$, and $\Theta$.

Finally, we can put everything together and prove
\begin{thm}[Second order is macroscopic locality]
\label{thm:second_order_is_ML}
In the $(2,2,2,2)$ scenario, $\set{IC}^\mathrm{2nd}(2,2,2,2)=\set{Q}^1(2,2,2,2)$.
\end{thm}

\begin{proof}
Let us first remove the deterministic case. If one of the marginals 
is deterministic, say $|\langle a_0\rangle|=1$, the behavior is 
local, since $p(a_1)\,p(b_0|a_1)\,p(b_1|a_1)$, together with the
deterministic value of $a_0$, is a joint distribution reproducing it; the other
cases are analogous. Local behaviors are quantum, hence they belong to $\set{Q}^1$ and to
$\set{IC}^{(1)}\subseteq\set{IC}^\mathrm{2nd}$. Let us therefore assume nondegenerate marginals.

($\subseteq$) If $P\in\set{IC}^\mathrm{2nd}$, Corollary~\ref{coro:ic2_characterization} and
Proposition~\ref{prop:second_order_2222} give $L(t)\succeq0$ for all $t$, hence
Eq.~\eqref{eq:tlm_landau}. In this scenario $\set{Q}^1$ is characterized by the
TLM condition on the renormalized correlators~\cite{navascues:2007,navascues:2010}, so
$P\in\set{Q}^1$.

($\supseteq$) Let $P\in\set{Q}^1$, so that $D(P)$ satisfies TLM. Then, there is a
quantum behavior $P'$ with zero marginals and
$\langle a_xb_y\rangle_{P'}=D_{x,y}(P)$~\cite{landau:1988,masanes:2003}; for it
$D(P')=D(P)$. Since quantum behaviors satisfy generalized IC,
$P'\in\set{IC}^\mathrm{2nd}$, i.e.\ $K_f(P',q)\succeq0$ for all $(n,d)$, $q$ and $f$. By
Lemma~\ref{lemma:second_order_structure}, in the coordinates
$(\alpha,\beta)$ the forms of $P$ and $P'$ coincide. The change of coordinates
\eqref{eq:alpha_beta_split} is invertible, hence $K_f(P,q)\succeq0$ as well, and
$P\in\set{IC}^\mathrm{2nd}$ by Corollary~\ref{coro:ic2_characterization}.
\end{proof}

\begin{coro}
\label{coro:ic_in_ml}
In the CHSH scenario, every behavior satisfying single-copy generalized IC
is macroscopically local,
$\set{IC}^{(1)}(2,2,2,2)\subseteq\set{Q}^1(2,2,2,2)$.
\end{coro}

Corollary~\ref{coro:ic_in_ml} recovers the result of
Ref.~\cite{gachechiladze:2026}. Within any expansion of the type in Eq.~\eqref{eq:general_channel},
the second-order constraints reproduce the
TLM criterion with no assumption on the database distribution, the strategies or the channel,
and Theorem~\ref{thm:second_order_is_ML} shows that no single-copy protocol
improves on them at this order. Constraints beyond $\set{Q}^1$ must therefore
come from beyond this second-order regime; that such constraints exist follows
from the macroscopically local behavior violating generalized IC with a
noiseless channel given in Ref.~\cite{gachechiladze:2026}.
In the next subsection, we show that perturbations changing the channel support 
can witness such violations already at first order.

\subsection{Perturbations changing the channel support}
\label{sec:support_change}
Here, we consider the complementary case to the previous one, where 
we allow $\mu(m')=0$ for some $m'$ and we allow perturbations that change 
the support of the channel, meaning that it might happen that $\mu(m')=0$ but
$(T_M)_{m',(\rr,a)} \ne 0$ for some $\rr,a$, once the channel is perturbed. 
Notice that this cannot happen in expansions like Eq.~\eqref{eq:general_channel}.

We show that such expansions can witness \emph{any} single-copy IC violation 
already at leading order, which is linear in $\varepsilon$. To see this, fix
a distribution $q$, deterministic choices $f_X,f_Y$ and a nonsignaling behavior 
$P$, and consider the dependence of the IC margin $\Lambda$ of Eq.~\eqref{eq:margin_def} 
on the channel $T_M$. As usual, $\Lambda(T_M) < 0$ implies a violation of IC.
Now, consider the IC task played with the channel $T_M^{\varepsilon}$ given by
composing $T_M$ with an erasure channel that discards the input with probability
$1-\varepsilon$. The discarded outcome is denoted by $\bot$, thus
\beq
\begin{aligned}
(T_M^{\varepsilon})_{\bot,(\rr,a)} &= 1 - \varepsilon  &&\text{ for all } \ \rr, a, \\
(T_M^{\varepsilon})_{m',(\rr,a)} &= \varepsilon (T_M)_{m', (\rr,a)}  &&\text{ for all } \ m', \rr, a.
\end{aligned}
\eeq 
If $M'$ can take $d'$ different values, the new channel has an output variable $M''$ 
taking $d'+1$ different values. Now we show that 
\beq
\Lambda(T_M^\varepsilon) = \varepsilon \Lambda(T_M). 
\eeq
To prove this, let us introduce a new variable $E$ defined as $e(m'')=1-\delta_{m'',\bot}$, 
it is $1$ if $M''\neq \bot$ and $0$ otherwise.
One can easily verify that $\Pr(E=1|\rr,a)= \sum_{m'} \varepsilon [T_M]_{m',(\rr,a)}=\varepsilon$.
By construction $E$ is independent of $(\mathbf R,A)$. Conditioning on $E=1$ corresponds to the use
of the original channel $T_M$. In particular, 
for $S=s$,
\begin{equation}
\begin{aligned}
p(\rr,a,m',b | E=1) = q(\rr) &p(a,b|x_{\rr},f_Y(s,m'))\\ 
&\qquad \qquad \times [T_M]_{m',(\rr,a)}.
\end{aligned}
\end{equation}
Instead, for $E=0$ the channel outputs $M''=\bot$ deterministically and
\begin{equation}
\begin{split}
p(\rr, b| E=0)&=\sum_ a p(\rr, a, b | E=0) \\
&= \sum_ a q(\rr) p(a, b |x_{\rr}, f_Y(s,\bot))\\
&= q(\rr) p(b |f_Y(s,\bot)),
\end{split}
\end{equation}
where Bob's input $f_Y(s,\bot)$ can be chosen arbitrarily for each $s$. 
Since $E$ is a deterministic function of $M''$, the chain rule gives
\begin{equation}
\begin{split}
I(\mathbf R,A:M'')
&=I(\mathbf R,A:M'' \mid E)\\
&=(1-\varepsilon) I(\mathbf R,A:M'' \mid E=0) \\ & \qquad \quad  + \varepsilon I(\mathbf R,A:M'' \mid E=1) \\ 
&= \varepsilon I_{T_M}(\mathbf R,A:M') ,
\end{split}
\end{equation}
In the first step, we also used $I(\mathbf R,A:E)=0$. Then, we averaged 
over $E$, and lastly we used the form of the channel given above for 
$E=0$ and $E=1$, denoting by $I_{T_M}(\mathbf{R}, A:M')$ the 
mutual information one obtains with the original protocol with $T_M$. 
The same chain rule can be applied to the other terms of the IC functional. 
For each $s$,
\begin{equation}
\begin{aligned}
    &I(R_s : M'',B_s \mid \mathbf R_{< s})  \\ 
    &\quad \ =I(R_s:E \mid \mathbf R_{<s}) + I(R_s:M'',B_s \mid E,\mathbf R_{<s}).
\end{aligned}
\end{equation}
The first term is zero because $E$ is independent of $\mathbf R$. 
The second term can be written again averaging over $E$ as
\begin{equation}
\begin{aligned}
(1-\varepsilon)  &I( R_s:M'',B_s \mid \mathbf R_{<s}, E=0)  \\ 
&\qquad \quad \quad + \varepsilon I(R_s:M'',B_s \mid \mathbf R_{<s}, E=1).
\end{aligned}
\end{equation}
The term with $E=0$ is again 0. Indeed, $E$ is independent of $\mathbf R$, so 
conditioning does not change the distribution
of $\mathbf{R}$. Moreover, when $E=0$ the channel gives the constant output $\bot$, and 
together with nonsignaling this makes
$(M'',B_s)$ independent of $\mathbf{R}$. 
In the last term with $E=1$, the channel acts as the original $T_M$, namely,
\begin{equation}
\begin{aligned}
&I(R_s:M'',B_s\mid\mathbf R_{<s},E=1)\\
&\qquad \qquad \qquad \qquad=I_{T_M}(R_s:M',B_s\mid\mathbf R_{<s}).
\end{aligned}
\end{equation}
Putting everything together, we obtain
\begin{equation}
\begin{aligned}
\Lambda(T_M^\varepsilon) &=\varepsilon \Big[ I_{T_M}(\mathbf R,A:M') \\
&\qquad \qquad -\sum_s I_{T_M}(R_s:M',B_s\mid\mathbf R_{<s}) \Big], \\
\end{aligned}
\end{equation}
which equals $\varepsilon\Lambda(T_M)$ by definition. By an analogous argument, the capacity
of $T_M^\varepsilon$ is rescaled by a factor $\varepsilon$.

Therefore, if $T_M$ violates IC, then $T_M^\varepsilon$ also violates it 
with the same gap rescaled by $\varepsilon \in (0,1)$. For 
$\varepsilon \to 0$ this channel converges to the zero-capacity channel that
outputs $\bot$ deterministically, $\mu(m'')=\delta_{m'',\bot}$.
As the single-copy set $\set{IC}^{(1)}$ is defined intersecting (among others) over 
all $d'$, we conclude that any single-copy IC violation can be detected arbitrarily 
close to a zero-capacity channel, at first-order in the perturbation $\varepsilon$. 

More concretely, consider the counterexample in Ref.~\cite{gachechiladze:2026} consisting
of a point $P \in \set Q^1 \setminus \set{IC}^{(1)}$. By the argument of this section,
composing the identity channel considered there with an erasure channel with parameter 
$\varepsilon \to 0$, one obtains an IC violation arbitrarily close to a zero-capacity
channel. Furthermore, by construction, this violation is not detected by the TLM inequality
as $P$ belongs to $\set Q^1$. We conclude that there exist, already at first order
in the expansion around zero-capacity channels, stronger constraints than the TLM 
inequality. An ingredient to find them is to consider perturbations that change
the channel support. In particular, considering a center of the series expansion that is
not the maximally noisy channel $\mu(m')=1/d'$, but a zero-capacity channel with constant
output $\mu(m'') = \delta_{m'',\bot}$ one can map any single-copy IC violation to another protocol
that violates IC and whose channel has arbitrarily small capacity.

\section{Discussion}

We have formulated the problem of determining whether a prescribed set of
nonsignaling correlations satisfies information causality as an optimization
over distributions of the sender's database, communication channels, and
deterministic local strategies. We showed that the optimal single-copy IC
deficit (i.e.~the maximum value of the difference between the two sides of the
IC inequality) is convex in the behavior shared by the parties, implying that
single-copy IC sets are convex. We characterized general multi-copy adaptive
protocols and showed that they admit an equivalent description as single-copy
protocols applied to effective behaviors obtained by local wirings of the
nonlocal resources. For sets of correlations closed under these wirings,
satisfying single-copy information causality across all Bell
scenarios therefore excludes any multi-copy violation. This applies to the
relevant case of the relation between the almost-quantum
set~\cite{navascues:2015} and information causality.

Furthermore, we have delineated a strategy to approximate the optimization
problems with analytical and numerical methods. As an application, we showed
that when employing the expansion of Ref.~\cite{jain:2024} around the noisy channel with full-support
output,
the second-order constraints from \emph{all} single-copy protocols
characterize exactly the macroscopically local set $\set Q^1$ in the CHSH
scenario, whereas considering first-order expansions around constant-output
zero-capacity channels, with perturbations that change the channel support,
already yields the same constraints as single-copy 
information causality itself. These results answer a question posed in
Ref.~\cite{gachechiladze:2026}. 

Several open problems remain. In the CHSH scenario, the IC violation by a 
macroscopically local behavior in Ref.~\cite{gachechiladze:2026} is not 
detected by the second-order constraints of full-support expansions, 
although, by Sec.~\ref{sec:support_change}, it persists at arbitrarily small capacity. 
This motivates a systematic study of which higher-order terms of full-support 
expansions give stronger constraints than macroscopic locality, and for which protocols this happens. Moreover, it is natural to
ask how the second-order IC set relates to macroscopic locality in larger Bell
scenarios. A further challenge is to obtain nontrivial upper bounds from
polynomial relaxations at manageable computational cost, in order to solve the
proposed membership and set-inclusion problems for relevant sets of correlations
such as polytopic sets~\cite{allcock:2009closed, lang:2014} and the almost-quantum
set~\cite{navascues:2015}.
Finally, we intuitively expect most of the results presented here, such as 
convexity properties and deterministic reductions, to apply to the recently introduced 
extension of IC to arbitrary distributed computing tasks \cite{miklin:2026cc, jain:2026extended}. 
We leave this problem for future investigations.

\section*{Acknowledgements}

We thank Mariami Gachechiladze, Nikolai Miklin, and Mirjam Weilenmann for insightful discussions.
MP~is supported by the STeP2 grant (ANR-22-EXES-0013) of the Agence Nationale de la Recherche (ANR) as 
part of Plan France 2030.

{\it AI Disclosure Statement.--} The conceptualization and development of the 
results on deterministic reduction, convexity, single-copy sufficiency, polynomial approximations, were 
performed solely by the authors. The authors acknowledge the use of
LLM tools (Claude Opus 5.5) for assistance in the expansion around zero-capacity channels. Moreover, the model exposed a mistake in a previous version of the manuscript and highlighted the role of support-changing perturbations. Finally, LLMs were used for improving grammar and syntax
of the text. The authors have verified all arguments in the final manuscript and take full responsibility for its contents.
\bibliography{bibliography}{}

\appendix

\begin{widetext}

\section{Proof of the single-copy deterministic reduction}
\label{app:generalized_deterministic_proofs}

    In this appendix, we prove Lemma~\ref{lemma:generalized_det_strat}.
    We recall that for any two random variables $U$ and $V$, the mutual
        information $I(U:V)$ is convex in $p(v|u)$ when $p(u)$ is fixed, and
        concave in $p(u)$ when $p(v|u)$ is fixed.
    Moreover,
    column-stochastic matrices of fixed size form a product of probability
    simplices, hence a convex polytope whose vertices are precisely the
    deterministic maps. A continuous convex function on such a polytope has a
    maximizer at an extreme point.
    We prove the reduction for every fixed database distribution
        $q \in \set{D}_{n,d}$.
    As in the main text, $G_s$ denotes the guess $G$ when $S = s$, and
        this conditioning is implicit whenever we write $G_s$.

    \begin{proof}[Deterministic $T_X$]
        Fix a database distribution $q \in \set{D}_{n,d}$ and compatible
        maps $T_M,T_Y,T_G$, and consider
        \begin{equation}
            f_q(T_X) := \Delta_1(P,q,T_X,T_M,T_Y,T_G)
        \end{equation}
        as a function of Alice's measurement strategy. We first consider a term
        on the left-hand side. By definition of the conditional mutual
        information,
        \begin{equation}
            I(R_s:G_s\,|\,\mathbf R_{<s})
=\sum_{\rr_{<s}}q(\rr_{<s})
I(R_s:G_s\,|\,\mathbf R_{<s}=\rr_{<s}). \label{eq:conditional_mi_decomposition_tx}
        \end{equation}
        For every $\rr_{<s}$ with $q(\rr_{<s}) > 0$, the
        conditional distribution $q(r_s|\rr_{<s})$ is fixed. The channel
        from $R_s$ to $G$, conditioned on
        $\mathbf R_{<s} = \rr_{<s}$ and $S = s$, is obtained by
        averaging over $\mathbf R_{>s}$ and all internal protocol
        variables. Explicitly,
        \begin{equation}
                \begin{split}
                    p(g|r_s,\rr_{<s},S = s) = \sum_{\rr_{>s},x,a,m',y,b} q(\rr_{>s}|r_s,\rr_{<s}) [T_X]_{x,\rr}\,p(a,b|x,y) [T_M]_{m',(\rr,x,a)}[T_Y]_{y,(s,m')} [T_G]_{g,(s,m',y,b)}.
                \end{split}
                \label{eq:induced_channel_tx}
            \end{equation}
        This conditional distribution is affine in the entries of $T_X$. Using
        the fact that mutual information is convex in the conditional channel
        for a fixed input distribution, each term in
        Eq.~\eqref{eq:conditional_mi_decomposition_tx}, and therefore their
        weighted sum, is convex in $T_X$. Values of $\rr_{<s}$ with
        $q(\rr_{<s}) = 0$ give no contribution.
        Conditional probabilities for input values of zero probability
            may be chosen arbitrarily, since their contributions to the mutual
            information vanish.

        For the right-hand side, fix $T_M$. The channel from
            $(\mathbf R,X,A)$ to $M'$ is fixed and equal to
        $p(m'|\rr,x,a) = [T_M]_{m',(\rr,x,a)}$. The corresponding input
        distribution is
        \begin{equation}
            p(\rr,x,a) = q(\rr)[T_X]_{x,\rr}p(a|x), \label{eq:input_distribution_tx}
        \end{equation}
        where $p(a|x)$ is the nonsignaling marginal of $P$.
        Equation~\eqref{eq:input_distribution_tx} is affine in $T_X$. Using the
        fact that, for a fixed channel, mutual information is concave in the
        input distribution, we have that
        \begin{equation}
            T_X \mapsto I(\mathbf R,X,A:M')
        \end{equation}
        is concave. Consequently, its negative contribution to $f_q(T_X)$ is
        convex.

        Combining both parts, $f_q(T_X)$ is convex on the polytope of
        column-stochastic matrices $T_X$. Its maximum can therefore be attained
        at a vertex, corresponding to a deterministic function $x = f_X(\rr)$.
        This conclusion holds for every fixed $q,T_M,T_Y,T_G$. Hence,
        taking the maximum over these variables afterwards,
        \begin{equation}
            \max_{q,T_X,T_M,T_Y,T_G}\Delta_1
            = \max_{q,T_M,T_Y,T_G}\max_{f_X}\Delta_1.
        \end{equation}
        In particular, the joint optimization over $q$ does not alter the
        deterministic reduction of $T_X$.
    \end{proof}

    \begin{proof}[Deterministic $T_Y$]
        Fix $q,T_X,T_M,T_G$ and consider
        \begin{equation}
            g_q(T_Y) := \Delta_1(P,q,T_X,T_M,T_Y,T_G).
        \end{equation}
        For each $s$ and each fixed $\rr_{<s}$, the conditional
            distribution of $R_s$ given $\mathbf R_{<s} = \rr_{<s}$ is fixed by
            the database distribution $q$.
        The induced conditional channel $R_s \to G_s$
        is affine in $T_Y$, since $T_Y$ appears linearly in
            Eq.~\eqref{eq:induced_channel_tx}, with $T_G$ fixed.
        It follows from convexity of mutual information in the channel that
        \begin{equation}
            T_Y \mapsto I(R_s:G_s\,|\,\mathbf R_{<s})
        \end{equation}
        is convex. Summing over $s$ preserves this convexity.

        The right-hand side $I(\mathbf R,X,A:M')$ depends only on the
        database distribution, Alice's measurement choice, the behavior
        through Alice's marginal, and the communication channel. It is therefore
        independent of $T_Y$. Thus $g_q(T_Y)$ is convex in $T_Y$, and a maximum
        is attained at an extreme point, corresponding to a deterministic
        function $y = f_Y(s,m')$. Again the argument is valid for every fixed
        $q,T_X,T_M,T_G$, so maximizing over $q$ afterwards leaves the
        conclusion unchanged.
    \end{proof}

    \begin{proof}[Deterministic $T_G$]
        For fixed $q,T_X,T_M,T_Y$, the channel in
            Eq.~\eqref{eq:induced_channel_tx} is affine in $T_G$, while the
            right-hand side of IC is independent of $T_G$. Thus the IC deficit
            is convex in $T_G$ and attains its maximum at a deterministic
            decoding function. Applying this after the two measurement
            reductions proves Lemma~\ref{lemma:generalized_det_strat}.
    \end{proof}

    The convexity statements above concern one stochastic protocol map at
        a time, while the database distribution $q$ is held fixed. If $q$ is
        varied, both the input distributions entering the conditional mutual
        informations and the distribution of Alice's variables $(\mathbf R,X,A)$
        change. The preceding argument therefore does not imply convexity or
        concavity of the IC deficit in $q$. Similarly, for $T_M$ the preceding
        convexity argument does not apply: for fixed $q,T_X,T_Y,T_G$, both the
        left- and the right-hand side of IC are convex in $T_M$, but their
        difference need not be. This is consistent with the results
        of Ref.~\cite{miklin:2021}, who obtain tighter single-copy IC bounds using
        noisy, non-extremal communication channels. In particular, the Tsirelson
        bound is recovered as the capacity of a binary symmetric channel tends
        to zero.

\section{Adaptive generalized multi-copy protocols}
\label{app:generalized_many_copy_protocols}

    In this appendix, we give the explicit parametrization of the adaptive
    protocols of
    Definition~\ref{dfn:sequential_adaptive_protocol}, discuss acyclicity and
    the role of order adaptivity on the two sides, and show, following
    Ref.~\cite{Beigi2015}, that the joint distribution is normalized and that
    the effective behavior of
    Eq.~\eqref{eq:generalized_adaptive_effective_behavior} is nonsignaling for
    any of these strategies. With respect to the standard concatenated
    protocols, the database has an arbitrary distribution $q$,
    and both parties may choose the order in which they use the copies
    adaptively.

    \subsection{Parametrization}
    \label{app:generalized_many_copy_parametrization}

    The parties share $c$ independent copies of $P$, with inputs and outputs
    labelled by copy:
    \begin{equation}
        \begin{split}
            \mathbf X = (X_1,\ldots,X_c),\qquad \mathbf Y
                      = (Y_1,\ldots,Y_c), \qquad \mathbf A
                      = (A_1,\ldots,A_c),\qquad \mathbf B
                      = (B_1,\ldots,B_c).
        \end{split}
        \label{eq:generalized_many_copy_registers}
    \end{equation}
    The local time step is denoted by $t \in \{1,\dots,c\}$. At her step $t$,
    Alice chooses to use the $\ell_t$-th box, and similarly at his step $t$ Bob
    chooses the $j_t$-th box. The tuples $\mathbf L,\mathbf J$ collect the
    orders of use of Alice and Bob, and corresponding lowercase symbols
    $\boldsymbol\ell$, $\mathbf j$ denote realized values. To describe the
    protocol explicitly, we follow the four-part division in
    Definition~\ref{dfn:sequential_adaptive_protocol}.

    \begin{enumerate}
        \item \emph{Alice's measurements.} At step $t$, Alice chooses which copy
              $\ell_t$ to use and which input $x_{\ell_t}$ to feed it according
              to the stochastic map
              \begin{equation}
                  p_A^{(t)}\left( \ell_t,x_{\ell_t} \middle| \rr, (\ell_u,x_{\ell_u},a_{\ell_u})_{u<t} \right), \label{eq:generalized_alice_tree}
              \end{equation}
              where
              $\ell_t \in \{1,\ldots,c\} \setminus \{\ell_1,\ldots,\ell_{t-1}\}$
              is chosen among the yet unused copies. The map is zero for a copy
              already used. After observing $a_{\ell_t}$, she includes this
              outcome and her choice in the variables available at the next
              step.
        \item \emph{Communication.} After all $c$ steps by Alice, she sends all
              her local information through a noisy classical channel
                  $p_M(m'|\rr,\boldsymbol\ell,\mathbf x,\mathbf a)$, 
              so that Bob receives $M'$.
        \item \emph{Bob's measurements.} Bob combines the received message $m'$
              with $s$ and all his local information. His choices at step $t$
              are specified by
              \begin{equation}
                  p_B^{(t)}\left( j_t,y_{j_t}\middle|s,m', (j_u,y_{j_u},b_{j_u})_{u<t} \right), \label{eq:generalized_bob_tree}
              \end{equation}
              where $j_t \in \{1,\ldots,c\} \setminus \{j_1,\ldots,j_{t-1}\}$ is
              chosen among unused copies, like above.
        \item \emph{Decoding.} After the last measurement, Bob generates $g$
              using all his information, through the map
                  $p_G(g|s,m',\mathbf j,\mathbf y,\mathbf b)$.
    \end{enumerate}

    Including the previous choices as well as the outcomes ensures that each
    party knows which copies remain unused, even for stochastic strategies. Note that
    for each fixed choice of $s,m',\mathbf j,\mathbf y,\mathbf b$, the distribution
    $p_G(\cdot|s,m',\mathbf j,\mathbf y,\mathbf b)$ can be chosen independently 
    of all the others. 
    The two orderings $\boldsymbol{\ell}$ and $\mathbf j$ are permutations of
    $1,\ldots,c$, so that each selected input is assigned to exactly one copy.
    The protocol generates the distribution
    \begin{equation}
        \begin{aligned}
            p(\rr,s,\boldsymbol\ell,\mathbf x,\mathbf a,m',\mathbf j,\mathbf y,\mathbf b,g) & = q(\rr)p(s)\,p_M(m'|\rr,\boldsymbol\ell,\mathbf x,\mathbf a) p_G(g|s,m',\mathbf j,\mathbf y,\mathbf b) \prod_{i=1}^c p(a_i,b_i|x_i,y_i)                                                                                                           \\
                                                                                            & \quad \times \prod_{t=1}^c \left[ p_A^{(t)}\left( \ell_t,x_{\ell_t} \middle| \rr, (\ell_u,x_{\ell_u},a_{\ell_u})_{u<t} \right) \right] \prod_{t=1}^c \left[ p_B^{(t)}\left( j_t,y_{j_t}\middle|s,m', (j_u,y_{j_u},b_{j_u})_{u<t} \right) \right] .
        \end{aligned}
        \label{eq:generalized_many_copy_joint}
    \end{equation}
    Lemma~\ref{lem:sequential_marginals} shows, following Ref.~\cite{Beigi2015},
    that this is the normalized distribution generated by the protocol for every
    nonsignaling $P$, including different adaptive orderings on the two sides.
    It also shows that Alice's variables and $M'$ have a marginal distribution
    independent of Bob's strategy and decoding map; the latter disappears upon
    summing over $g$. Notice also that if any copy is not used by a party in the
    protocol, that party can measure it last and ignore the outcome. By
    nonsignaling, this leaves the IC deficit unchanged.
    Let $T_A^{(t)},T_B^{(t)}$ be the column-stochastic matrices representing
    $p_A^{(t)},p_B^{(t)}$, respectively, let $T_M$ represent the communication
    channel, and $T_G$ denote the matrix representing $p_G$. Then, the
    operational part of the protocol is
    \begin{equation}
        \pi_c = \left( \{T_A^{(t)}\}_{t=1}^c,T_M,\{T_B^{(t)}\}_{t=1}^c,T_G \right). \label{eq:generalized_many_copy_protocol_tuple}
    \end{equation}
    Together with the database distribution $q \in \set{D}_{n,d}$, these maps
    fully specify the protocol.

    \subsection{Normalization and nonsignaling}
    \label{app:generalized_effective_ns}

    The measurement steps of the protocols of
    Definition~\ref{dfn:sequential_adaptive_protocol} are analogous to the
    general wirings of Ref.~\cite[Sec.~3]{Beigi2015}, in which each party
    chooses the next box and its input by a stochastic map of its complete
    locally available information. Here, the copies
    are identical (even though considering different boxes essentially leaves
    the argument unchanged), the effective outputs are the complete outputs
    $\mathbf a$, $\mathbf b$, and the channel producing $m'$ acts between the
    two parties' decision trees, with $m'$ entering in all steps of Bob's
    decision tree. The properties we need follow from the argument used in
    Ref.~\cite[App.~C]{Beigi2015} to derive the joint distribution of a wiring,
    which also covers these modifications. 
    In particular, one might think that
    different adaptive orderings on the two sides could lead to circular
    definitions of the internal inputs. For two copies, Alice may choose $x_2$
    using $a_1$, while Bob chooses $y_1$ using $b_2$, so $x_2$ might seem to
    depend on itself through the two copies. However, nonsignaling allows Alice
    to obtain $a_1$ without waiting for $y_1$. Moreover, in these protocols,
    each party's choices depend only on information already available locally,
    and communication goes only from Alice to Bob through $M'$, before Bob
    starts his measurements. These dependencies therefore form an acyclic graph.
    The following lemma establishes normalization of the resulting joint
    distribution for arbitrary adaptive orderings.

    \begin{lemma}[Joint distribution and marginals]
        \label{lem:sequential_marginals}
        Let $P$ be nonsignaling and consider a protocol of
        Definition~\ref{dfn:sequential_adaptive_protocol}, with possibly
        stochastic decision trees. Then:
        \begin{enumerate}
            \item the protocol generates the joint distribution in
                  Eq.~\eqref{eq:generalized_many_copy_joint}, which is
                  normalized;
            \item the marginal distribution of
                  $(\mathbf R,\mathbf L,\mathbf X,\mathbf A,M')$ is
                  \begin{equation}
                      q(\rr)\,p_M(m'|\rr,\boldsymbol\ell,\mathbf x,\mathbf a) \prod_{t=1}^c p_A^{(t)}\left( \ell_t,x_{\ell_t}\middle|\rr, (\ell_u,x_{\ell_u},a_{\ell_u})_{u<t} \right) \prod_{i=1}^c p(a_i|x_i), \label{eq:generalized_alice_marginal}
                  \end{equation}
                  and does not depend on Bob's measurement operations or
                  decoding map;
            \item for deterministic decision trees, the effective behavior $P'$
                  of Eq.~\eqref{eq:generalized_adaptive_effective_behavior},
                  with internal inputs $x_i = x_i(x',\mathbf a)$ and
                  $y_i = y_i(y',\mathbf b)$ as in
                  Eq.~\eqref{eq:alice_internal_inputs}, is a normalized
                  nonsignaling behavior with marginals
                  \begin{align}
                      p'(a'|x') & = \prod_{i=1}^c p(a_i|x_i(x',\mathbf a)), \label{eq:generalized_effective_ns_alice} \\
                      p'(b'|y') & = \prod_{i=1}^c p(b_i|y_i(y',\mathbf b)). \label{eq:generalized_effective_ns_bob}
                  \end{align}
        \end{enumerate}
    \end{lemma}

    \begin{proof}
        We follow Ref.~\cite[App.~C]{Beigi2015}, letting the copies be possibly
        different behaviors $P_1,\ldots,P_c$. If every $P_i$ is a product
        behavior, $p_i(a,b|x,y) = p_{i}(a|x)\,p_{i}(b|y)$, the two halves of
        each copy are independent local devices and the protocol, including the
        channel and final decoding map, is a classical stochastic process. Its
        distribution is the product of the conditional probabilities of the
        individual steps, and relabelling the box factors by copy rather than by
        step, which is possible since $t \mapsto \ell_t$ and $t \mapsto j_t$ are
        permutations of $\{1,\ldots,c\}$, gives
        Eq.~\eqref{eq:generalized_many_copy_joint}. For such behaviors,
        Eqs.~\eqref{eq:generalized_alice_marginal}--\eqref{eq:generalized_effective_ns_bob} follow by summing out the
        local process of the other party, whose total probability is one.

        For a fixed protocol $\pi_c$ and general nonsignaling behaviors, the
        distributions in Eqs.~\eqref{eq:generalized_many_copy_joint}
        and~\eqref{eq:generalized_alice_marginal}--\eqref{eq:generalized_effective_ns_bob}
        are linear in each $P_i$, since each copy is used exactly once by each
        party. Moreover, every nonsignaling behavior lies in the linear span of
        product behaviors~\cite{Beigi2015}. Hence, applying this argument one
        copy at a time, the identities above hold for all nonsignaling
        $P_1,\ldots,P_c$, and in particular for $P_i = P$. Normalization is
        preserved, since a decomposition $P = \sum_k\mu_kP^{(k)}$ into product
        behaviors satisfies $\sum_k\mu_k = 1$ by normalization of $P$. The
        right-hand sides of Eqs.~\eqref{eq:generalized_effective_ns_alice}
        and~\eqref{eq:generalized_effective_ns_bob} do not depend on $y'$ and
        $x'$, respectively, so $P'$ is nonsignaling.
    \end{proof}

    The argument does not use the order in which the copies are used by either
    party, but only that each party uses every copy once and chooses its inputs
    from its own local data.
    
\subsection{Closure of $\set{Q}^1$ and $\set{AQ}$ under sequential adaptive wirings}
\label{app:closure_adaptive_wirings}

In this subsection, we verify that the two families used in
Sec.~\ref{sec:generalized_many_copy_ic} satisfy the closure property of
Eq.~\eqref{eq:generalized_identical_copy_closure} for the wirings of
Definition~\ref{dfn:sequential_adaptive_protocol}. Throughout, the decision
trees are deterministic, as allowed by
Lemma~\ref{lem:generalized_deterministic_many_copy}, and the effective
outputs $a' = \mathbf a$ and $b' = \mathbf b$ contain all local outputs
labelled by copy.

Recall that a nonsignaling behavior $P$ belongs to $\set{Q}^1$, the set of
\emph{macroscopically local} correlations~\cite{navascues:2010}, if and only if
there exists a positive-semidefinite matrix $\Gamma$ with rows and columns
labelled by the \emph{null event} $\phi$ and by the events $(a|x)$ and
$(b|y)$, whose entries satisfy $\Gamma_{\phi\phi} = 1$,
$\Gamma_{\phi,(a|x)} = p(a|x)$, $\Gamma_{\phi,(b|y)} = p(b|y)$,
$\Gamma_{(a|x),(b|y)} = p(a,b|x,y)$,
$\Gamma_{(a|x),(a'|x)} = \delta_{aa'}p(a|x)$ and
$\Gamma_{(b|y),(b'|y)} = \delta_{bb'}p(b|y)$, the remaining entries being
unconstrained~\cite{navascues:2007, navascues:2008, navascues:2010}. By the
Gram decomposition, this is equivalent to the existence of real vectors
$|\phi\rangle$, $\{|a,x\rangle\}$ and $\{|b,y\rangle\}$ in some real inner
product space with
\begin{equation}
    \begin{gathered}
        \langle\phi|\phi\rangle = 1,\qquad \langle\phi|a,x\rangle = p(a|x),\qquad \langle\phi|b,y\rangle = p(b|y),                                        \\
        \langle a,x|a',x\rangle = \delta_{aa'}\,p(a|x),\qquad \langle b,y|b',y\rangle = \delta_{bb'}\,p(b|y), \qquad \langle a,x|b,y\rangle = p(a,b|x,y).
    \end{gathered}
    \label{eq:q1_gram_conditions}
\end{equation}
If the certificate is written only for nonzero outcomes, as in the
standard formulation of the NPA hierarchy, the vectors for the outcome $0$
can be recovered as $|0,x\rangle:=|\phi\rangle-\sum_{a\neq0}|a,x\rangle$ and
$|0,y\rangle:=|\phi\rangle-\sum_{b\neq0}|b,y\rangle$. Normalization and
nonsignaling then give Eq.~\eqref{eq:q1_gram_conditions} for all outcomes.
For instance, $\langle0,x|b,y\rangle=p(b|y)-\sum_{a\neq0}p(a,b|x,y)=p(0,b|x,y)$.

\begin{lemma}[Closure of $\set{Q}^1$ under sequential adaptive wirings]
    \label{lem:q1_closure_adaptive}
    Let $P \in \set{Q}^1(\mathsf B)$, let $c < \infty$, and consider a
    deterministic adaptive wiring for Alice and Bob as in
    Definition~\ref{dfn:sequential_adaptive_protocol}, with arbitrary
    orderings on both sides. Then the effective behavior $P'$ of
    Eq.~\eqref{eq:generalized_adaptive_effective_behavior} belongs to
    $\set{Q}^1(\mathsf B')$.
\end{lemma}

\begin{proof}
The argument follows the proof of closure under wirings in
Ref.~\cite{navascues:2010}. Let
$|\phi\rangle,\{|a,x\rangle\},\{|b,y\rangle\}$ satisfy
Eq.~\eqref{eq:q1_gram_conditions} for $P$, and take one copy
$\mathcal V_i$ of the corresponding inner-product space for every
$i\in \{ 1,\dots,c \}$. On $\bigotimes_{i=1}^c\mathcal V_i$ define
\begin{equation}
    |\Phi\rangle \coloneqq \bigotimes_{i=1}^c|\phi\rangle_i, \qquad |a',x'\rangle \coloneqq \bigotimes_{i=1}^c|a_i,x_i(x',\mathbf a)\rangle_i, \qquad |b',y'\rangle \coloneqq \bigotimes_{i=1}^c|b_i,y_i(y',\mathbf b)\rangle_i, \label{eq:q1_tensor_vectors}
\end{equation}
where $a'=\mathbf a$, $b'=\mathbf b$, and $x_i(x',\mathbf a)$,
$y_i(y',\mathbf b)$ are the internal inputs generated by the two trees, as
in Eq.~\eqref{eq:alice_internal_inputs}. Inner products of tensor-product
vectors factorize over the copies, so that
\begin{equation}
\begin{aligned}
\langle\Phi|\Phi\rangle&=1,\\
\langle\Phi|a',x'\rangle&=\prod_{i=1}^c p(a_i|x_i(x',\mathbf a))=p'(a'|x'),\\
\langle\Phi|b',y'\rangle&=\prod_{i=1}^c p(b_i|y_i(y',\mathbf b))=p'(b'|y'),\\
\langle a',x'|b',y'\rangle&=\prod_{i=1}^c p(a_i,b_i|x_i,y_i)
=p'(a',b'|x',y'),
\end{aligned}
\end{equation}
where we used Eqs.~\eqref{eq:generalized_effective_ns_alice},
\eqref{eq:generalized_effective_ns_bob}
and~\eqref{eq:generalized_adaptive_effective_behavior}. For $a'=a''$, the
same factorization gives
$\langle a',x'|a',x'\rangle=\prod_ip(a_i|x_i)=p'(a'|x')$.

It remains to show that $\langle a',x'|a'',x'\rangle = 0$
whenever $a' \neq a''$. For
the two output lists $a'$ and $a''$, the first copy ($\ell_1$) and
input ($x_{\ell_1}$) are the same, since both are functions of $x'$. As
long as the observed outputs agree, 
Alice continues to choose the same copies and inputs. Since
$a' \neq a''$ and every copy is used once, there must be a first
step $t$ at which the outputs differ. The copy $\ell_t$ chosen at
that step has the same input in both cases,
$x_{\ell_t}(x',a') = x_{\ell_t}(x',a'')$, but
$a'_{\ell_t} \neq a''_{\ell_t}$. The corresponding factor in the
inner product is zero by Eq.~\eqref{eq:q1_gram_conditions}, so the
full inner product vanishes. The same argument applies to Bob's
vectors.

The vectors in Eq.~\eqref{eq:q1_tensor_vectors} therefore satisfy
Eq.~\eqref{eq:q1_gram_conditions} for the effective behavior $P'$, which
is nonsignaling by Lemma~\ref{lem:sequential_marginals}. Their Gram matrix
is positive semidefinite, and hence
$P'\in\set{Q}^1(\mathsf B')$.
\end{proof}

The proof uses the fact that the next copy and input are the same whenever the effective input and all
previous outputs agree. This holds for every deterministic decision tree,
even when the measurement order depends on earlier outputs, as observed in
Ref.~\cite{navascues:2010}.

Let us now consider the case of the almost-quantum set. Closure of
$\set{AQ}$ under wirings is proved in
Ref.~\cite[Supplementary Note~2]{navascues:2015}, by combining the closure
under composition of independent boxes with the closure under grouping of
parties followed by a local wiring. In our setting, the $c$ copies of
$P \in \set{AQ}$ are first composed into a $2c$-partite almost-quantum
behavior, and then the $c$ parties on Alice's side and the $c$ parties on
Bob's side are grouped in turn. Grouping can be iterated because the
behavior obtained after the first grouping is again almost-quantum, and
hence again admits a representation in terms of a state and projectors. The
grouping argument of Ref.~\cite{navascues:2015} represents each outcome of a
fine-grained wiring, i.e., a wiring that records all local outputs as our
effective outputs do, by the product of the projectors corresponding to the
inputs and outputs of the grouped parties. It then uses two facts. First, the order
of this product does not affect its action on the almost-quantum state, since products of
projectors of different parties acting on that state are invariant under
permutations of the parties. This is the defining property of almost-quantum
correlations for products involving all parties, and it extends to products
over subsets of parties by inserting resolutions of the identity for the
missing ones. Second, two different outcomes of the same fine-grained wiring
involve some party that was measured with the same input and gave different
outputs, which is the property
established in the proof of Lemma~\ref{lem:q1_closure_adaptive}. Both facts
hold even when the measurement order depends on earlier outputs. The proof
of Ref.~\cite{navascues:2015} therefore applies without modification, giving
$P' \in \set{AQ}(\mathsf B')$.

    \section{Proof of deterministic generalized multi-copy strategies}
    \label{app:generalized_deterministic_many_copy}

    Here, we prove Lemma~\ref{lem:generalized_deterministic_many_copy}. As in
    Appendix~\ref{app:generalized_deterministic_proofs}, we keep the
    database distribution $q$ fixed while extremizing one stochastic
    protocol map and maximize over $q$ only after all local reductions have been
    performed.

    \begin{proof}[Proof of Lemma~\ref{lem:generalized_deterministic_many_copy}]
        We use two facts about the distribution of all variables in
        Eq.~\eqref{eq:generalized_many_copy_joint}. Each map
        $T_A^{(t)}$ or $T_B^{(t)}$ appears once
        in the product, so the joint distribution is affine in its
        entries
        when all other maps are fixed. Moreover, the probabilities
            assigned to different lists of previous choices and outcomes
        are independent
        parameters. The admissible maps therefore form
        a product of probability simplices, restricted to unused
        copies, whose vertices are deterministic maps.

        Fix $q \in \set{D}_{n,d}$ and all maps except $T_A^{(t)}$, for
        some $t \in \{1,\ldots,c\}$.
        For fixed $\mathbf R_{<s} = \rr_{<s}$ and $S = s$, the
        input distribution of $R_s$ is fixed by $q$.
        With $T_G$ held fixed, the
        induced channel from $R_s$ to $G_s$ is affine in
        $T_A^{(t)}$. Convexity of mutual information in the
        channel implies the convexity of
        \begin{equation}
                T_A^{(t)} \mapsto I(R_s:G_s \mid \mathbf R_{<s} = \mathbf r_{<s}).
            \end{equation}
        Averaging over $\mathbf R_{<s}$ with the fixed weights from $q$
            and summing over $s$
        preserves this convexity.

        For fixed $T_M$, the channel from
        $(\mathbf R,\mathbf L,\mathbf X,\mathbf A)$
        to $M'$ is fixed. Its input distribution, obtained by summing
        Eq.~\eqref{eq:generalized_alice_marginal} over $m'$, is affine in
        $T_A^{(t)}$.
        Thus $I(\mathbf R,\mathbf L,\mathbf X,\mathbf A:M')$ is concave
        in $T_A^{(t)}$, and $\Delta_c$ is convex in that map.
        Consequently, its maximum is attained at a vertex,
         corresponding to a deterministic map.
        Repeating this for $t = 1,\ldots,c$ makes Alice's
        entire measurement strategy deterministic without decreasing the
        IC deficit $\Delta_c$.

        For a Bob map $T_B^{(t)}$, the induced channels from
        $R_s$ to $G_s$ are affine in
        $T_B^{(t)}$ with $T_G$ fixed, so the left-hand side of IC is
            convex.
        By Lemma~\ref{lem:sequential_marginals}, the marginal of
        $(\mathbf R,\mathbf L,\mathbf X,\mathbf A,M')$ is independent of
        Bob's maps. Hence the communication term is
        unchanged, and $\Delta_c$ is convex in
        $T_B^{(t)}$ as well. Choosing a vertex at each step makes Bob's
            entire measurement strategy deterministic.

        Finally, fix $q,T_M$ and all $T_A^{(t)}, T_B^{(t)}$. The conditional
            channels from $R_s$ to $G_s$ are affine in $T_G$, while the
            right-hand side of IC is independent of $T_G$. Thus $\Delta_c$ is
            convex in $T_G$ and attains its maximum at a deterministic decoding
            map.

        We conclude that
        \begin{equation}
            \max_{q,\pi_c}\Delta_c(P,q,\pi_c)
            = \max_{q,\pi_c}\left\{\Delta_c(P,q,\pi_c)
                \, \middle|\,  \{T_A^{(t)}\}_t,\{T_B^{(t)}\}_t,T_G\ \mathrm{deterministic}\right\}.
                \label{eq:generalized_many_copy_pointwise_q_reduction}
        \end{equation}
        The maximization on the right still includes generic $q$ and
            $T_M$, but all other maps can be taken
            deterministic.
    \end{proof}

    Notice that the deterministic reduction we proved here applies to the
        choice of the next copy as well as to its input.
    In other words, the deterministic map $T_A^{(t)}$ selects an unused copy and
        its input through two deterministic functions:
        \beq
            \ell_t = f_{L_t}(\rr,a_{\ell_1},\ldots,a_{\ell_{t-1}}), \qquad x_{\ell_t}
                   = f_{X_t}(\rr,a_{\ell_1},\ldots,a_{\ell_{t-1}}).
        \eeq
        Previous copy choices $(\ell_1, \dots, \ell_{t-1})$ and inputs
        $(x_{\ell_1}, \dots, x_{\ell_{t-1}})$
        need not appear as arguments of these two functions, as they
        can be recovered by applying the earlier functions to $\rr$ and the
        outcomes already observed. Similarly, Bob's deterministic choices may be written as
        \beq
            j_t = f_{J_t}(s,m',b_{j_1},\ldots,b_{j_{t-1}}), \qquad
            y_{j_t} = f_{Y_t}(s,m',b_{j_1},\ldots,b_{j_{t-1}}).
        \eeq

\section{Proof of convexity of generalized single-copy sets}
\label{app:generalized_single_copy_conv}

Here, we prove Proposition~\ref{prop:generalized_convexity_single_copy_ic}.
The proof of convexity is pointwise both in the database distribution $q$
and in the protocol $\pi_1$. This makes explicit why allowing the optimizing
$q$ to depend on the behavior, similarly to Ref.~\cite{gachechiladze:2026},
preserves convexity. As in Sec.~\ref{sec:generalized_ic_sets}, we may take
$\set{D}_{n,d}$ to be the full probability simplex on $[d]^n$ (although the
proof also applies to any fixed set of database distributions independent of
$P$).

\begin{proof}[Proof of Proposition~\ref{prop:generalized_convexity_single_copy_ic}]
Fix first a distribution $q \in \set{D}_{n,d}$ and a reduced single-copy
protocol $\pi_1 = (f_X,T_M,f_Y,f_G)$. The IC deficit is
\begin{equation}
    \begin{split}
        \Delta_1(P,q,\pi_1) = \sum_{s=0}^{n-1}I(R_s:G_s \mid \mathbf R_{<s})
                              - I(\mathbf R,A:M').
    \end{split}
    \label{eq:generalized_single_copy_delta_convexity}
\end{equation}
We show that this is convex in $P$. Consider one term of the sum, and
condition on $\mathbf R_{<s} = \rr_{<s}$. For every value with
$q(\rr_{<s}) > 0$, the input distribution $q(r_s|\rr_{<s})$ is fixed.
The corresponding transition probability from $R_s$ to $G_s$ is
\begin{equation}
    \begin{split}
        p(g|r_s,\rr_{<s},S = s) = \sum_{\rr_{>s},a,m',b}q(\rr_{>s}|r_s,\rr_{<s}) [T_M]_{m',(\rr,a)} p\!\left( a,b\middle| f_X(\rr),f_Y(s,m') \right)\delta_{g,f_G(s,m',b)}.
    \end{split}
    \label{eq:generalized_convexity_induced_channel}
\end{equation}
This conditional distribution is affine in the entries of $P$. Since
mutual information is convex in the conditional distribution for a fixed
input distribution,
$P \mapsto I(R_s:G_s \mid \mathbf R_{<s} = \rr_{<s})$ is convex.
Averaging with the fixed weights $q(\rr_{<s})$ shows that
$P \mapsto I(R_s:G_s \mid \mathbf R_{<s})$ is also convex.

For the last term of
Eq.~\eqref{eq:generalized_single_copy_delta_convexity}, consider the
random variable $(\mathbf R,A)$. Its input distribution is
\begin{equation}
    p(\rr,a) = q(\rr) p\!\left( a|f_X(\rr) \right), \label{eq:generalized_convexity_alice_input}
\end{equation}
which is affine in $P$, while the transition probability from
$(\mathbf R,A)$ to $M'$ is fixed by $T_M$,
$p(m'|\rr,a) = [T_M]_{m',(\rr,a)}$. Mutual information is concave in the
input distribution for a fixed conditional distribution. Hence
$P \mapsto I(\mathbf R,A:M')$ is concave. Its negative contribution to
Eq.~\eqref{eq:generalized_single_copy_delta_convexity} is therefore
convex. Combining the two parts, the function
\begin{equation}
    P \mapsto \Delta_1(P,q,\pi_1)
\end{equation}
is convex for every fixed pair $(q,\pi_1)$.
        
Taking the pointwise supremum over protocols preserves convexity, so
$P \mapsto \Delta_{n,d,d'}^{*,(1)}(P;q)$ is convex for every fixed $q$.
Taking the pointwise supremum once more over $q \in \set{D}_{n,d}$ shows
that $P \mapsto \Delta_{n,d,d'}^{*,(1)}(P)$ is also convex.
Equivalently, for $P_\lambda = \lambda P_0 + (1 - \lambda)P_1$ with
$0 \leq \lambda \leq 1$,
\begin{equation}
    \begin{split}
        \Delta_{n,d,d'}^{*,(1)} (P_\lambda;\set{D}_{n,d}) & = \sup_{q,\pi_1}\Delta_1(P_\lambda,q,\pi_1)                                                           \\
                                                          & \leq \lambda\sup_{q,\pi_1} \Delta_1(P_0,q,\pi_1) + (1 - \lambda)\sup_{q,\pi_1} \Delta_1(P_1,q,\pi_1).
    \end{split}
    \label{eq:generalized_convexity_direct_inequality}
\end{equation}
The optimizing $(q,\pi_1)$ on the left-hand side need not coincide with
an optimizer for either $P_0$ or $P_1$. The single-copy IC set with RAC
parameters $(n,d,d')$,
\begin{equation}
    \begin{split}
        \set{IC}_{n,d,d'}^{(1)} (\mathsf B;\set{D}_{n,d}) = \{P \in \set{NS}(\mathsf B) \mid \Delta_{n,d,d'}^{*,(1)} (P;\set{D}_{n,d}) \leq 0\},
    \end{split}
\end{equation}
is a sublevel set of a convex function, intersected with the convex
nonsignaling set, hence it is convex. Finally, the full single-copy IC
set $\set{IC}^{(1)}(\mathsf B)$ is convex as it is an intersection of
such convex sets over $(n,d,d')$.
\end{proof}

We stress again that this proof is specific to the single-copy case. For a
fixed $c$-copy protocol with $c > 1$, the joint probabilities generated by
the protocol are multilinear in the $c$ copies and become polynomial of
degree $c$ after the copies are identified with the same behavior $P$.
Therefore, the affinity in $P$ used in the proof need not hold for
multi-copy protocols.

\section{Proof of the second-order form}
\label{app:second_order_form}

Here, we provide a proof of Proposition~\ref{prop:second_order_form}. The main idea
is to work in the $s$th \emph{branch}, i.e., conditioning on $S=s$, and use Lemma~\ref{lemma:mi_near_zero} for each term of Eq.~\eqref{eq:margin_def}.

\begin{proof}[Proof of Proposition~\ref{prop:second_order_form}]
We start with Alice's term $I(\mathbf R,A:M')$. From Lemma~\ref{lemma:mi_near_zero}, 
with $V=(\mathbf R,A)$, $Z=M'$, base distribution $\mu$ and
perturbation $u_{\rr,a}=w_{\rr,a}$, which is centered by
Eq.~\eqref{eq:general_channel}, we get
\begin{equation}
I(\mathbf R,A:M')=\frac{\varepsilon^2}2 \sum_{m'}\mu(m')\Var\bigl(w_{\mathbf R,A}(m')\bigr)+O(\varepsilon^3).
\end{equation}

For Bob's terms $I(R_s:M',B_s \mid \mathbf R_{<s})$, we start by noticing that, with his setting 
$y^s_{m'}=f_Y(s,m')$ fixed after receiving $m'$, the joint distribution of $(\rr,a,m',b)$ is
$q(\rr)\,[T_M]_{m',(\rr,a)}\,p(a,b|x_\rr,y^s_{m'})$, which gives us 
Eq.~\eqref{eq:bob_transcript}, as discussed in the main text. The base distribution
$\mu_s(m',b)=\mu(m')p(b|y^s_{m'})$, then, does not depend on $\rr$, and the perturbation
is centered, since 
$\sum_b\mu_s(m',b)\,\zeta^{(s)}_\rr(m',b)=\mu(m')\sum_ap(a|x_\rr)\,w_{\rr,a}(m')$ and
$\sum_{m'}\mu(m')w_{\rr,a}(m')=0$. We note that outcomes with $p(b|y)=0$ do not occur and are
just omitted. 

Conditioning on a prefix $\mathbf R_{<l}$ averages the perturbation over the
remaining components of $\mathbf R$, i.e., 
$p(m',b|\mathbf r_{<l})=\mu_s(m',b)\bigl[1+ \varepsilon \,\EE[\zeta^{(s)}_{\mathbf R}(m',b)|\mathbf r_{<l}]\bigr]$.
By the chain rule,
$I(R_s:M',B_s \mid \mathbf R_{<s})=I(\mathbf R_{\leq s}:M',B_s)-I(\mathbf R_{<s}:M',B_s)$, and
Lemma~\ref{lemma:mi_near_zero} applied to $V=\mathbf R_{\leq s}$ and $V=\mathbf R_{<s}$
gives
\begin{equation}
\begin{split}
I(R_s:M',B_s \mid \mathbf R_{<s})= \frac{\varepsilon^2}2 \sum_{m'}\mu(m')\sum_bp(b|y^s_{m'})
\Bigl[\Var\bigl(\EE[\zeta^{(s)}_{\mathbf R}(m',b)|\mathbf R_{\leq s}]\bigr)
-\Var\bigl(\EE[\zeta^{(s)}_{\mathbf R}(m',b)|\mathbf R_{<s}]\bigr)\Bigr]+O(\varepsilon^3).
\end{split}
\end{equation}

We now collect all the terms, grouping them by the value $m'$ of the message. For a
fixed $m'$, Bob uses in branch $s$ the setting $f_Y(s,m')$, exactly as in the
fixed-message strategy $f_Y^{m'}$. Hence Alice's term and Bob's terms with
this $m'$ are those of Eq.~\eqref{eq:hessian_matrix} for the strategy $f_Y^{m'}$,
evaluated at the direction $w(m')$. Summing over $m'$ with weights $\mu(m')$ gives
Eq.~\eqref{eq:second_order_form}. The expansion is valid because, for small
$\varepsilon$, all channel probabilities stay close to $\mu(m')>0$, 
while probabilities that vanish at $\varepsilon=0$ remain zero, so every
entropy in the margin is a smooth function of $\varepsilon$.
\end{proof}

\section{Second-order form for binary outcomes}
\label{app:centered_outcomes}

Here, we prove Lemma~\ref{lemma:second_order_structure} that gives the second-order
form of the IC functional in the case of binary outcomes $a,b \in \{0,1\}$, using
the properties of the centered and normalized outcomes $\xi_x$, $\eta_y$ and showing 
that terms involving $\alpha$ cancel.

\begin{proof}[Proof of Lemma~\ref{lemma:second_order_structure}]
Let's start with Alice's term. Since $\EE[\xi_{x_\rr}(A)|\rr]=0$  the cross term $\EE[\alpha_{\mathbf R}\beta_{\mathbf R}\xi_{x_{\mathbf R}}(A)]$ is zero,  and using
$\EE[\xi_{x_\rr}(A)^2|\rr]=1$ we are left with
\begin{equation}
\Var(w_{\mathbf R,A})=\Var(\alpha_{\mathbf R})
+\sum_\rr q(\rr)\beta_\rr^2.
\end{equation}
Now, fix a branch $s$ and write $y=y_s$. By
Eq.~\eqref{eq:hessian_matrix}, Bob's term in this branch is
\begin{equation}
\sum_bp(b|y)\Bigl[\Var\bigl(\EE[\zeta^{(s)}_{\mathbf R}(b)|\mathbf R_{\leq s}]\bigr)
-\Var\bigl(\EE[\zeta^{(s)}_{\mathbf R}(b)|\mathbf R_{<s}]\bigr)\Bigr],
\end{equation}
i.e., the difference between two levels of knowledge of the string: the
prefix $\mathbf R_{<l}$ with $l=s+1$ and with $l=s$. Consider a generic prefix length
$l \in \{0,\dots,n\}$. By Eq.~\eqref{eq:zeta_centered}, conditioning on $\mathbf R_{<l}$ averages
$\alpha_{\mathbf R}$ and $\beta_{\mathbf R}D_{x_{\mathbf R},y}$ over the components
of the string that are not observed, while $\eta_y(b)$ is fixed:
\begin{equation}
\EE[\zeta^{(s)}_{\mathbf R}(b)|\mathbf R_{<l}]
=\EE[\alpha_{\mathbf R}|\mathbf R_{<l}]+\eta_y(b)\,\EE[g^{(s)}_{\mathbf R}|\mathbf R_{<l}],
\qquad g^{(s)}_\rr:=\beta_\rr D_{x_\rr,y}.
\end{equation}
For fixed $b$, the variance over $\mathbf R$ of a sum $X+\eta_y(b)\,Y$ is
$\Var(X)+2\eta_y(b)\operatorname{Cov}(X,Y)+\eta_y(b)^2\Var(Y)$. Averaging over $b$
with weights $p(b|y)$, the cross term vanishes, since $\sum_bp(b|y)\eta_y(b)=0$,
and the last term keeps its coefficient, since $\sum_bp(b|y)\eta_y(b)^2=1$.
Therefore
\begin{equation}
\sum_bp(b|y)\Var\bigl(\EE[\zeta^{(s)}_{\mathbf R}(b)|\mathbf R_{<l}]\bigr)
=V_l+\Var\bigl(\EE[g^{(s)}_{\mathbf R}|\mathbf R_{<l}]\bigr),
\qquad V_l:=\Var\bigl(\EE[\alpha_{\mathbf R}|\mathbf R_{<l}]\bigr).
\label{eq:bob_split}
\end{equation}
Bob's term in branch $s$ is the difference of Eq.~\eqref{eq:bob_split} at
$l=s+1$ and $l=s$. It thus splits into an $\alpha$ part, $V_{s+1}-V_s$, and a
$\beta$ part, involving $g^{(s)}$ only.
Summing the $\alpha$ parts over the branches, consecutive terms cancel,
\begin{equation}
\sum_{s=0}^{n-1}\bigl(V_{s+1}-V_s\bigr)=V_n-V_0=\Var(\alpha_{\mathbf R}),
\end{equation}
since $V_0=\Var(\EE[\alpha_{\mathbf R}])=0$, the mean being a constant, and
$V_n=\Var(\EE[\alpha_{\mathbf R}|\mathbf R])=\Var(\alpha_{\mathbf R})$. This
equals the $\alpha$ part of Alice's term, and the two cancel in
Eq.~\eqref{eq:hessian_matrix}.
By the tower property of conditional expectations, 
i.e., $\EE\bigl[\EE[X|H]\bigr]=\EE[X]$, we get that
$\EE[g^{(s)}_{\mathbf R}|\mathbf R_{\leq s}]$ and $\EE[g^{(s)}_{\mathbf R}|\mathbf R_{<s}]$ 
have the same mean $\EE[g^{(s)}_{\mathbf R}]$, so the difference of 
their variances equals the difference of their second moments,
$\EE\bigl[\EE[g^{(s)}_{\mathbf R}|\mathbf R_{\leq s}]^2\bigr]
-\EE\bigl[\EE[g^{(s)}_{\mathbf R}|\mathbf R_{<s}]^2\bigr]$. Together with
Alice's $\beta$ part, $\sum_\rr q(\rr)\beta_\rr^2$, this gives
Eq.~\eqref{eq:phi2_beta}. The behavior enters only through the coefficients
$D_{x,y}$ in $g^{(s)}$.
\end{proof}

\section{Reduction of the second-order condition in the CHSH scenario}
\label{app:reduction_2222}

Here we prove Proposition~\ref{prop:second_order_2222}. The main idea is to go 
through all possible choices for Bob's strategy, the database distribution $q(\rr)$, 
and Alice's strategy, yielding the condition $L(t) \succeq 0$.

\begin{proof}[Proof of Proposition~\ref{prop:second_order_2222}]
With $n=2$, Eq.~\eqref{eq:phi2_beta} reads
\begin{equation}
\Phi_2=\sum_{i,j}q(i,j)\beta_{ij}^2
-\Bigl\{\EE\bigl[\EE[g^{(0)}|R_0]^2\bigr]-\EE[g^{(0)}]^2\Bigr\}
-\Bigl\{\EE\bigl[(g^{(1)})^2\bigr]-\EE\bigl[\EE[g^{(1)}|R_0]^2\bigr]\Bigr\},
\label{eq:phi2_2222}
\end{equation}
with $g^{(s)}_{ij}=\beta_{ij}D_{x_{ij},f_Y(s)}$ and $x_{ij}=f_X(i,j)$.

Let us look at all the options for Bob's strategy. If $f_Y(0)=f_Y(1)=y$, then 
$g^{(0)}=g^{(1)}=:g$. In the two braces the terms $\EE[\EE[g^{(0)}|R_0]^2\bigr]=\EE[\EE[g^{(1)}|R_0]^2\bigr]=\EE[\EE[g|R_0]^2\bigr]$ cancel out, leaving $\EE[g]^2-\EE\bigl[g^2\bigr]$. Finally, writing $\EE\bigl[g^2\bigr]= \sum_{i,j}q(i,j)\beta_{ij}^2 D_{x_{ij},y}^2$, we obtain
\begin{equation}
\Phi_2=\sum_{i,j}q(i,j)\beta_{ij}^2\bigl(1-D_{x_{ij},y}^2\bigr)+\EE[g]^2\geq0
\end{equation}
for every behavior since $|D_{x_{ij},y}|\leq 1$. If $f_Y(s)=1-s$, relabeling Bob's settings reduces to
$f_Y(s)=s$, as discussed below. It is, then, enough to assume $f_Y(s)=s$, which we do from now on.

Let us focus on the database distribution. For each row $i$ of $\rr = (i,j)$ and each branch $s$, let
\begin{equation}
\bar g^{(s)}_i:=\EE\bigl[g^{(s)}_{\mathbf R}\big|R_0=i\bigr]
=\sum_j\tau_i(j)\, g_{ij}^{(s)}=\sum_j\tau_i(j)\,\beta_{ij}\,D_{x_{ij},s}
\label{eq:row_average}
\end{equation}
be the average of $g^{(s)}$ within row $i$, since $q(i,j)=p_i \tau_i(j)$. 
We now evaluate the three terms of Eq.~\eqref{eq:phi2_2222}. Alice's term splits over the rows,
$\sum_{i,j}q(i,j)\beta_{ij}^2=\sum_ip_i\sum_j\tau_i(j)\beta_{ij}^2$. The first
brace comes from the branch $s=0$, in which Bob learns the row, i.e., index $i$. It can be interpreted as 
the variance of $\bar g^{(0)}_{R_0}$, a random variable taking the values $\bar g^{(0)}_0$ and
$\bar g^{(0)}_1$ with probabilities $p$ and $1-p$, namely,
\begin{equation}
\begin{split}
\EE\bigl[\EE[g^{(0)}|R_0]^2\bigr]-\EE[g^{(0)}]^2
&=p\,(\bar g^{(0)}_0)^2+(1-p)(\bar g^{(0)}_1)^2-\bigl(p\,\bar g^{(0)}_0+(1-p)\bar g^{(0)}_1\bigr)^2\\
&=p(1-p)\bigl(\bar g^{(0)}_0-\bar g^{(0)}_1\bigr)^2,
\end{split}\label{eq:C_0-C_1}
\end{equation}
where we used that $\EE[g^{(0)}]=\sum_{ij} p_i \tau_i(j) g_{ij}^{(0)}=p\,\bar g^{(0)}_0+(1-p)\bar g^{(0)}_1$.

The second brace comes from the branch $s=1$, in which Bob learns the position
within the row, i.e., index $j$. 
It is the average over the rows of the variance within each row,
\begin{equation}
\EE\bigl[(g^{(1)})^2\bigr]-\EE\bigl[\EE[g^{(1)}|R_0]^2\bigr]
=\sum_ip_i\Bigl[\sum_j\tau_i(j)\bigl(g^{(1)}_{ij}\bigr)^2-\bigl(\bar g^{(1)}_i\bigr)^2\Bigr].
\label{eq:A_i}
\end{equation}
Using $(g^{(1)}_{ij})^2=\beta_{ij}^2D_{x_{ij},1}^2$ and collecting the terms of each
row gives
\begin{equation}\label{eq:Phi_2_A_C}
\Phi_2=p\,A_0+(1-p)\,A_1-p(1-p)\,(C_0-C_1)^2,
\end{equation}
 with
\begin{equation}
C_i=\bar g^{(0)}_i,
\qquad
A_i=\sum_j\tau_i(j)\,\beta_{ij}^2\bigl(1-D_{x_{ij},1}^2\bigr)+\bigl(\bar g^{(1)}_i\bigr)^2 ,
\label{eq:AC_def}
\end{equation}
where the term $-p(1-p)\,(C_0-C_1)^2$ comes from Eq.~\eqref{eq:C_0-C_1} and the terms with $p_i A_i$ come from combining the first term $\sum_{i,j}q(i,j)\beta_{ij}^2$ with those in Eq.~\eqref{eq:A_i}. 
Note that the quantities $A_i$ and $C_i$ depend only on row $i$, i.e., on the map
$j\mapsto x_{ij}$, on $t_i$ and on $\beta_{i,0},\beta_{i,1}$; the parameter $p$
appears only through the explicit coefficients in Eq.~\eqref{eq:Phi_2_A_C}.

We now show that $\Phi_2\geq0$ for all $p\in[0,1]$ and all $\beta$ if and only if
$A_i\geq C_i^2$ for $i=0,1$. Suppose that $A_i\geq C_i^2$ for both rows. Then
\begin{equation}
\begin{split}
\Phi_2&\geq p\,C_0^2+(1-p)C_1^2-p(1-p)(C_0-C_1)^2\\
&=p^2C_0^2+2p(1-p)\,C_0C_1+(1-p)^2C_1^2=\bigl(p\,C_0+(1-p)C_1\bigr)^2\geq0 .
\end{split}
\end{equation}
Conversely, suppose that $\Phi_2\geq0$ for all $p$ and $\beta$. Setting
$\beta_{1,0}=\beta_{1,1}=0$ gives $\bar g^{(0)}_1=\bar g^{(1)}_1=0$, hence
$A_1=C_1=0$ and $\Phi_2=p\,\bigl[A_0-(1-p)C_0^2\bigr]$. For every $p\in(0,1]$ this
requires $A_0\geq(1-p)C_0^2$, and letting $p\to0$ gives $A_0\geq C_0^2$. Exchanging
the roles of the two rows gives $A_1\geq C_1^2$. Since the two rows involve
disjoint variables, the conditions decouple: $\Phi_2\geq0$ for all database distributions and all
$\beta$ if and only if, for each row, $A_i-C_i^2\geq0$ for all $t_i\in[0,1]$ and
all $\beta_{i,0},\beta_{i,1}$.

Let us now look at Alice's strategy. Consider one row fixed and drop the row index, writing
$\tau=(t,1-t)$, $\beta=(\beta_0,\beta_1)$ and $x_j=f_X(i,j)$. By
Eq.~\eqref{eq:AC_def},
\begin{equation}
A-C^2=\sum_j\tau_j\beta_j^2\bigl(1-D_{x_j,1}^2\bigr)
+\Bigl(\sum_j\tau_j\beta_jD_{x_j,1}\Bigr)^2
-\Bigl(\sum_j\tau_j\beta_jD_{x_j,0}\Bigr)^2 .
\label{eq:row_form}
\end{equation}
For a constant row, $x_0=x_1=x$, both squared sums are proportional to
$(\sum_j\tau_j\beta_j)^2$, and
\begin{equation}
A-C^2=\bigl(1-D_{x,1}^2\bigr)\sum_j\tau_j\beta_j^2
+\bigl(D_{x,1}^2-D_{x,0}^2\bigr)\Bigl(\sum_j\tau_j\beta_j\Bigr)^2 .
\end{equation}
Since $\sum_j\tau_j\beta_j^2\geq(\sum_j\tau_j\beta_j)^2$ and $1-D_{x,1}^2\geq0$, this
gives $A-C^2\geq(1-D_{x,0}^2)(\sum_j\tau_j\beta_j)^2\geq0$ for every behavior: rows
with a constant setting give no condition. For the row $x_j=j$, i.e., $x_0=0$ and
$x_1=1$, Eq.~\eqref{eq:row_form} is a quadratic form in $(\beta_0,\beta_1)$ with
coefficients
\begin{equation}
\begin{aligned}
\beta_0^2&:\quad t\bigl(1-D_{0,1}^2\bigr)+t^2\bigl(D_{0,1}^2-D_{0,0}^2\bigr)
=t\bigl[1-(1-t)D_{0,1}^2-t\,D_{0,0}^2\bigr],\\
\beta_1^2&:\quad (1-t)\bigl(1-D_{1,1}^2\bigr)+(1-t)^2\bigl(D_{1,1}^2-D_{1,0}^2\bigr)
=(1-t)\bigl[1-t\,D_{1,1}^2-(1-t)D_{1,0}^2\bigr],\\
\beta_0\beta_1&:\quad 2t(1-t)\bigl(D_{0,1}D_{1,1}-D_{0,0}D_{1,0}\bigr),
\end{aligned}
\end{equation}
so that $A-C^2=\beta^{\mathsf T}L(t)\,\beta$, with $L(t)$ as in
Eqs.~\eqref{eq:L_matrix} and~\eqref{eq:L_entries}. For the row $x_j=1-j$,
relabeling $j\to1-j$ exchanges $\tau_0$ and $\tau_1$, i.e., maps $t$ to $1-t$, and
returns the row $x_j=j$. It therefore gives the condition $L(1-t)\succeq0$, which
belongs to the same family.

Suppose that $L(t)\succeq0$ for all $t\in[0,1]$. Then every row satisfies
$A_i\geq C_i^2$: the rows with $x_j=j$ and $x_j=1-j$ by the computation above, and
the constant rows for every behavior. By the decoupling argument, $\Phi_2\geq0$ for
every database distribution. This covers the strategy $f_Y(s)=s$; constant $f_Y$
was treated above. The remaining strategy, $f_Y(s)=1-s$, is obtained from
$f_Y(s)=s$ by exchanging Bob's settings, $D_{x,0}\leftrightarrow D_{x,1}$, and gives
the conditions $\tilde L(t)\succeq0$, where $\tilde L(t)$ is $L(t)$ with $D_{x,0}$
and $D_{x,1}$ exchanged. Under this exchange $\ell_0(t)\to\ell_0(1-t)$,
$\ell_1(t)\to\ell_1(1-t)$ and $\kappa\to-\kappa$, so that, for $0<t<1$,
\begin{equation}
\tilde L(t)=M_t\,L(1-t)\,M_t,
\qquad
M_t=\operatorname{diag}\Bigl(\sqrt{\tfrac{t}{1-t}},\,-\sqrt{\tfrac{1-t}{t}}\Bigr).
\end{equation}
Since $M_t$ is invertible, $\tilde L(t)\succeq0$ if and only if $L(1-t)\succeq0$; at
$t\in\{0,1\}$ both matrices are diagonal with nonnegative entries. This strategy
therefore gives no new condition, and $K_f(P,q)\succeq0$ for all database
distributions and strategies. Conversely, if $K_f(P,q)\succeq0$ for all database
distributions and strategies, choose $f_Y(s)=s$, a strategy $f_X$ with the row
$x_{0j}=j$, and any $t_0\in[0,1]$. The decoupling argument gives
$\beta^{\mathsf T}L(t_0)\beta\geq0$ for all $\beta$, i.e., $L(t_0)\succeq0$. This
concludes the proof.
\end{proof}

\section{Proof of the Tsirelson--Landau--Masanes inequality}
\label{app:proof_tlm}

Here, we complete the proof of Proposition \ref{prop:tlm_landau}, namely, 
the derivation of the Tsirelson--Landau--Masanes inequality from the second
order expansion of single-copy IC. The only step left is the 
elimination of the parameter $t$. For $0<t<1$, $L(t)=D_t\widehat L(t)D_t$ with 
the invertible matrix $D_t=\operatorname{diag}(\sqrt t,\sqrt{1-t})$, and $\widehat L(t)$
given by
\begin{equation}
\widehat L(t) =
\begin{pmatrix}
\ell_0(t) & \sqrt{t(1-t)}\,\kappa\\
\sqrt{t(1-t)}\,\kappa & \ell_1(t)
\end{pmatrix}.
\label{eq:app_hat_L}
\end{equation}
It follows that $L(t)\succeq0$ if and only if $\widehat L(t)\succeq0$. 
The diagonal entries of $\widehat L(t)$,
\begin{equation}
\ell_0(t)=(1-t)\bigl(1-D_{0,1}^2\bigr)+t\bigl(1-D_{0,0}^2\bigr),
\qquad
\ell_1(t)=t\bigl(1-D_{1,1}^2\bigr)+(1-t)\bigl(1-D_{1,0}^2\bigr),
\end{equation}
are convex combinations of nonnegative numbers, and the off-diagonal entry is
$\sqrt{t(1-t)}\,\kappa$ with $\kappa=D_{0,1}D_{1,1}-D_{0,0}D_{1,0}$. A symmetric
$2\times2$ matrix with nonnegative diagonal entries is positive semidefinite if and
only if its determinant is nonnegative. The determinant
$\ell_0(t)\ell_1(t)-t(1-t)\kappa^2$ is a polynomial of degree two in $t$. At $t=0$ it
equals $(1-D_{0,1}^2)(1-D_{1,0}^2)=R$, and at $t=1$ it equals
$(1-D_{0,0}^2)(1-D_{1,1}^2)=S$. Its coefficient of $t^2$ is the product of the
slopes of $\ell_0$ and $\ell_1$, plus $\kappa^2$ from
$-t(1-t)\kappa^2=-t\kappa^2+t^2\kappa^2$,
\begin{equation}
\bigl(D_{0,1}^2-D_{0,0}^2\bigr)\bigl(D_{1,0}^2-D_{1,1}^2\bigr)
+\bigl(D_{0,1}D_{1,1}-D_{0,0}D_{1,0}\bigr)^2
=\bigl(D_{0,0}D_{1,1}-D_{0,1}D_{1,0}\bigr)^2=\Theta^2,
\end{equation}
since the terms $D_{0,1}^2D_{1,1}^2$ and $D_{0,0}^2D_{1,0}^2$ cancel. A polynomial
of degree two is determined by its values at $t=0$ and $t=1$ and by its leading
coefficient, which gives Eq.~\eqref{eq:det_hat_L}, namely,
$\det\widehat L(t)=(1-t)R+tS-t(1-t)\Theta^2$, a convex function of $t$. Dividing by
$t(1-t)>0$, $\det\widehat L(t)\geq0$ is equivalent to
\begin{equation}
\Theta^2\leq\frac Rt+\frac S{1-t},
\end{equation}
and this must hold for all $t\in(0,1)$. By the Cauchy--Schwarz inequality, applied
to the vectors $(\sqrt{R/t},\sqrt{S/(1-t)})$ and $(\sqrt t,\sqrt{1-t})$,
\begin{equation}
\Bigl(\frac Rt+\frac S{1-t}\Bigr)\bigl(t+(1-t)\bigr)\geq\bigl(\sqrt R+\sqrt S\bigr)^2,
\end{equation}
with equality if and only if $\sqrt R/t=\sqrt S/(1-t)$, i.e., at
$t^*=\sqrt R/(\sqrt R+\sqrt S)$ when $R,S>0$. If $R=0$ or $S=0$, the infimum
$(\sqrt R+\sqrt S)^2$ is approached as $t\to0$ or $t\to1$. Hence $L(t)\succeq0$ for
all $t\in(0,1)$ if and only if $|\Theta|\leq\sqrt R+\sqrt S$, which is
Eq.~\eqref{eq:tlm_landau}. At $t\in\{0,1\}$, $L(t)$ is diagonal with nonnegative
entries and gives no further condition.

\end{widetext}

\end{document}